\documentclass[10pt,a4paper,lefttitle]{article}
\usepackage[utf8]{inputenc}
\usepackage[english]{babel}
\usepackage[total={6in, 10in}]{geometry}
\usepackage{xparse}

\usepackage{amsmath,amssymb,amsthm,mathtools}
\usepackage{mleftright,bbm,accents}
\usepackage[autostyle=true]{csquotes}
\usepackage{enumitem}
\usepackage[font=small,format=hang,justification=justified,singlelinecheck=false]{caption}
\usepackage[font=small,format=hang,justification=justified,singlelinecheck=false]{subcaption}

\usepackage{xcolor}
\definecolor{tumblue}{RGB}{0,101,189}

\usepackage{tikz}
\usetikzlibrary{arrows.meta,calc,positioning,shapes,decorations.pathreplacing}
\tikzset{%
  vertex/.style={draw=black,circle,fill=black,text width=3pt,inner sep=0pt,outer sep=1pt},%
  arc/.style={-{Latex}}
}

\usepackage[colorlinks=true,urlcolor=tumblue,citecolor=tumblue,linkcolor=tumblue]{hyperref}
\def\EMAIL#1{\href{mailto:#1}{#1}}

\usepackage[capitalise,noabbrev,nameinlink]{cleveref}
\usepackage[backend=biber,style=authoryear-comp,citestyle=authoryear-comp,maxbibnames=99,maxcitenames=2, sortlocale=en_EN,natbib=true,url=false,doi=false,eprint=false,giveninits=true,uniquelist=false]{biblatex}

\makeatletter

\theoremstyle{definition}
\newtheorem{definition}{Definition}

\theoremstyle{plain}
\newtheorem{lemma}[definition]{Lemma}
\newtheorem{theorem}[definition]{Theorem}
\newtheorem{proposition}[definition]{Proposition}
\newtheorem{corollary}[definition]{Corollary}

\theoremstyle{remark}

\newtheorem{example}[definition]{Example}

\Crefname{definition}{Definition}{Definitions}
\crefname{definition}{definition}{definitions}
\Crefname{lemma}{Lemma}{Lemmas}
\crefname{lemma}{Lemma}{Lemmas}
\Crefname{theorem}{Theorem}{Theorems}
\crefname{theorem}{Theorem}{Theorems}
\Crefname{proposition}{Proposition}{Propositions}
\crefname{proposition}{Proposition}{Propositions}
\Crefname{corollary}{Corollary}{Corollaries}
\crefname{corollary}{Corollary}{Corollaries}
\Crefname{observation}{Observation}{Observations}
\crefname{observation}{Observation}{Observations}
\Crefname{remark}{Remark}{Remarks}
\crefname{remark}{Remark}{Remarks}
\Crefname{algocf}{Algorithm}{Algorithms}
\crefname{algocf}{Algorithm}{Algorithms}
\Crefname{algorithm}{Algorithm}{Algorithms}
\crefname{algorithm}{Algorithm}{Algorithms}
\Crefname{problem}{Problem}{Problems}
\crefname{problem}{Problem}{Problems}
\Crefname{figure}{Figure}{Figures}
\crefname{figure}{Figure}{Figures}

\let\Rn\relax
\let\d\relax
\ProvideDocumentCommand\N{}{\mathbb{N}}
\ProvideDocumentCommand\Z{}{\mathbb{Z}}
\ProvideDocumentCommand\Rnn{}{\mathbb{R}_{< 0}}
\ProvideDocumentCommand\Rn{}{\mathbb{R}_{\le 0}}
\ProvideDocumentCommand\R{}{\mathbb{R}}
\ProvideDocumentCommand\Rp{}{\mathbb{R}_{\ge 0}}
\ProvideDocumentCommand\Rpp{}{\mathbb{R}_{> 0}}
\ProvideDocumentCommand\one{}{\mathbbm{1}}
\ProvideDocumentCommand\Nu{}{N}
\ProvideDocumentCommand\d{}{\,\text{d}}
\ProvideDocumentCommand\doubleprime{}{{\prime\mkern-1mu\prime}}
\newcommand*\boldcdot{\mathpalette\boldcdot@{.75}}
\newcommand*\boldcdot@[2]{\mathbin{\vcenter{\hbox{\scalebox{#2}{$\m@th#1\bullet$}}}}}

\ProvideDocumentCommand\series{}{\mathbin{*}}
\ProvideDocumentCommand\parallel{}{\mathbin{\|}}
\ProvideDocumentCommand\cupdot{}{\mathbin{\mathaccent\cdot\cup}}

\DeclareMathOperator*{\essinf}{ess\,inf}
\DeclareMathOperator*{\esssup}{ess\,sup}
\DeclareMathOperator{\dom}{dom}

\DeclareMathOperator{\graph}{graph}

\DeclareMathOperator{\hyp}{hyp}
\DeclareMathOperator{\diag}{diag}
\DeclareMathOperator{\Id}{Id}

\ProvideDocumentCommand\ubar{m}{\underaccent{\bar}{#1}}

\NewDocumentCommand\xDeclarePairedDelimiter{mmm}
 {%
  \NewDocumentCommand#1{som}{%
   \IfNoValueTF{##2}
    {\IfBooleanTF{##1}{#2##3#3}{\mleft#2##3\mright#3}}
    {\mathopen{##2#2}##3\mathclose{##2#3}}%
  }%
 }

\xDeclarePairedDelimiter\floor{\lfloor}{\rfloor}
\xDeclarePairedDelimiter\ceil{\lceil}{\rceil}
\xDeclarePairedDelimiter\set{\lbrace}{\rbrace}
\xDeclarePairedDelimiter\abs{|}{|}
\xDeclarePairedDelimiter\card{|}{|}
\xDeclarePairedDelimiter\norm{\|}{\|}
\xDeclarePairedDelimiter\pos{[}{]_+}

\let\Det\det
\let\det\relax
\ProvideDocumentCommand\det{som}{%
  \Det\IfBooleanTF{#1}{#3}{\IfNoValueTF{#2}{\mleft(#3\mright)}{\mathopen{#2(}#3\mathclose{#2)}}}%
}%

\ProvideDocumentCommand\inv{m}{#1^{\raisebox{0.3ex}{$\scriptscriptstyle-\!1$}}}

\makeatother

\begin{document}

    \renewcommand{\arraystretch}{1.5}

    \title{%
      Computation of Dynamic Equilibria\\in Series-Parallel Networks%
      \thanks{This work has been supported by the Alexander von Humboldt Foundation with funds from the German Federal Ministry of Education and Research (BMBF).}%
    }

    \author{%
      Marcus Kaiser%
      \footnote{Operations Research, Technische Universität München, \EMAIL{marcus.kaiser@tum.de}}%
    }

    \maketitle

    \begin{abstract}
        We consider dynamic equilibria for flows over time under the fluid queuing model.
        In this model, queues on the links of a network take care of flow propagation.
        Flow enters the network at a single source and leaves at a single sink.
        In a dynamic equilibrium, every infinitesimally small flow particle reaches the sink as early as possible given the pattern of the rest of the flow.
        While this model has been examined for many decades, progress has been relatively recent.
        In particular, the derivatives of dynamic equilibria have been characterized as thin flows with resetting, which allowed for more structural results.

        Our two main results are based on the formulation of thin flows with resetting as linear complementarity problem and its analysis.
        We present a constructive proof of existence for dynamic equilibria if the inflow rate is right-monotone.
        The complexity of computing thin flows with resetting, which occurs as a subproblem in this method, is still open.
        We settle it for the class of two-terminal series-parallel networks by giving a recursive algorithm that solves the problem for all flow values simultaneously in polynomial time.
    \end{abstract}

    \section{Introduction}
    \label{section:introduction}

    Systems with large temporal fluctuation, as it emerges for example in road traffic, cannot be adequately captured by models based on static flows.
    This shortcoming is already addressed by \citet*{FF1962}, who look at flow that travels through a network at a (finite) speed.
    In their basic model for \emph{flows over time}, the network is defined on a directed graph.
    Each arc represents a link with some constant delay and capacity.
    Many optimization problems in this model are well-understood.
    \citet*{S2009} gives a good overview in his survey.

    The \emph{fluid queuing model}, as considered already by \citet*{V1969} in the context of transportation investments, builds upon this model and allows to study flow which exhibits selfish behavior.
    Each link in the network is equipped with a fluid queue at its entrance.
    If the capacity constraint for a link is violated, the excess flow is collected in its queue.
    Consequently, waiting in the queue imposes additional delay for flow to traverse a link.
    For a given inflow rate at the source, each infinitesimally small flow particle is considered a player in a \emph{routing game over time} who tries to reach the sink in the shortest possible time.
    Under this assumption, a \emph{dynamic equilibrium} describes a state of the system in which no player has an incentive to deviate.
    Each player is assumed to have full information.
    Hence, they anticipate the behavior of other flow particles and the state of the fluid queues at the respective moment they would reach the queues.
    As the queues respect the first-in-first-out principle, there cannot be overtaking of particles in a dynamic equilibrium.
    Therefore, the shortest path for a particle entering at the source is determined by all the flow which has entered before.

    \paragraph{Related work.}

    The problem at hand belongs to the class of \emph{dynamic traffic assignment problems}, as introduced by \citet*{MN1978a,MN1978b}.
    An overview of the topic is available by \citet*{PZ2001}.
    Traffic planners tend to base their decisions on complex simulations.
    The theory regarding the existence and computation of dynamic equilibria, however, has not caught up on the complexity of the models in use (see \citet*{CHOP2018}).

    Existence of dynamic equilibria for a class of models related to the one described above was shown by \citet*{ZM2000}.
    The work by \citet*{MW2010} concerns a general framework for dynamic congestion games.
    Their results regarding the existence of equilibria are the first to apply to the model as defined above.
    The involved methods, however, are non-constructive.
    The basis for a constructive proof of existence of dynamic equilibria and further insights into their nature is provided by \citet*{KS2011}.
    These authors establish a specific structure of the derivatives of dynamic equilibria which they called \emph{thin flows with resetting}.
    Based on this, they suggest a method to compute a dynamic equilibrium for constant inflow rates by integrating over thin flows with resetting.
    The integration is done in phases during which the derivative stays constant.
    One such step is called \emph{$\alpha$-extension}.

    \citet*{CCL2015} refine the notion of thin flows with resetting to \emph{normalized thin flows with resetting} and prove their existence and the uniqueness of the associated labels.
    This turns the $\alpha$-extension into a constructive proof of existence of dynamic equilibria for piecewise constant inflow rate.
    The authors augment this by also showing existence for inflow rates in $L^p(0, T)$, where $1 < p < \infty,\,T \ge 0,$ via variational inequalities (also for multiple origin-destination pairs).
    \citet*{SS2018} obtain an algorithmic approach for a multiple-source multiple-sink setting, in which the inflow rates at the sources are constant and the inflow of every source is routed to the sinks proportionally with respect to a global pattern.

    The model allows notions of a \emph{Price of Anarchy} with respect to various objectives for the social optimum.
    First results are given by \citet*{KS2011}.
    \citet*{BFA2015} analyze the Price of Anarchy in a Stackelberg game.
    By reducing the capacities of the links, a leader can ensure that a fixed amount of flow in a dynamic equilibrium reaches the sink within a factor of $\tfrac{e}{e - 1}$ compared to the time it would take in a social optimum of the original network.
    Further progress on this Price of Anarchy was made by \citet*{CCO2019}, who prove a bound of $\tfrac{e}{e - 1}$ under weaker assumptions, which get dispensable if a certain \emph{monotonicity conjecture} holds.

    \citet*{CCO2017} consider the long-term behavior of dynamic equilibria.
    They prove that for constant inflow rates the queues stay constant after a finite amount of time (given a necessary condition).
    At the same time, they give a small series-parallel network capturing distinctive phenomena of dynamic equilibria.

    Similar methods were successfully used on related models.
    \citet*{SVK2019} investigate a variant of the fluid queuing model in which the total amount of flow on a link at any point in time is bounded.
    If links fill up, congestion is propagated backwards across vertices to their ingoing links, which is called \emph{spillback}.
    \citet*{GH2019}, on the other hand, achieve existence results for a different concept of equilibria in the fluid queuing model which does not assume that players have full information.

    \paragraph{Our contribution.}

    We linearize a known non-linear complementarity problem by \citet*{CCL2011} and use it to examine normalized thin flows with resetting and their dependency on the flow value (which is determined by the inflow rate).
    This enables us to generalize the $\alpha$-extension to a larger class of inflow rates, namely right-monotone, locally integrable functions.
    Here, a function is \emph{right-monotone} if each point is the lower end point of a closed interval on which the function is monotone.
    We also use our new insights into the properties of normalized thin flows with resetting to tackle their computation.
    This subproblem in the $\alpha$-extension is known to be in PPAD~(see~\cite{P1994}), and it remains unclear whether it lies in P.
    We settle the computational complexity for the class of two-terminal series-parallel networks (see~\cite{BLS1999}) by giving a polynomial-time algorithm that solves the problem for all flow values simultaneously.

    \paragraph{Structure of this paper.}

    This paper is organized as follows.
    \Cref{section:model} introduces the fluid queuing model in a more formal manner.
    The most important known results for our purposes are stated.
    This is continued in \cref{section:equilibrium,section:normalizedthinflow} in which the notion of a dynamic equilibrium and normalized thin flows are defined, and known characterizations are given.
    In \cref{section:normalizedthinflow}, we further examine normalized thin flows with resetting parametrized by the flow value, which is needed in the subsequent two sections.
    \Cref{section:evolution} generalizes the known constructive proof of the existence of dynamic equilibria to be able to cope with right-monotone, locally integrable inflow rates.
    In \cref{section:seriesparallel}, we consider the computation of normalized thin flows with resetting.
    We prove that this computation can be done efficiently for two-terminal series-parallel networks.
    \cref{section:conclusion} finishes with some concluding remarks.

    \section{The model}
    \label{section:model}

    Throughout this paper, we consider a network which is given by a directed graph $G = (V, A)$, positive arc capacities $\nu \in \Rpp^A$, and non-negative delays $\tau \in \Rp^A$.
    We assume that there is no (directed) cycle $C \subseteq A$ in $G$ with zero total transit time, i.e.,  $\sum_{a \in C} \tau_a = 0$.
    For a set of vertices $U \subseteq V$, we denote the set of incoming and outgoing arcs by $\delta^-(U)$ and $\delta^+(U)$, respectively.
    Further, $\delta(U) := \delta^-(U) \cup \delta^+(U)$ are all arcs crossing the cut $U$.
    If $U$ consists only of a single vertex $u \in V$, we write $\delta^-(u)$ and $\delta^+(u)$ for $\delta^-(U)$ and $\delta^+(U)$.
    Closely related to that, $N^-(u) := \set{v \in V\colon (v, u) \in A}$ and $N^+(u) := \set{w \in V\colon (u, w) \in A}$ denote the sets of in-neighbors and out-neighbors of $u$, respectively.

    Let $s \in V$ be the source and $t \in V$ be the sink of a single commodity.
    Without loss of generality, we can assume that there is a (directed) $s$--$v$-path in $G$ for every vertex $v \in V$ (by removing other vertices).
    Further, the rate at which flow enters the network at the source~$s$ at any point in time is described by a non-negative function $\nu_0 \in L^{1}_{\text{loc}}(\R)$ that vanishes almost everywhere on $\Rnn$.
    Here, $L^{1}_{\text{loc}}(\R)$ denotes the set of \emph{locally (Lebesgue-)integrable functions} on $\R$. (As usual, we identify functions which are equal almost everywhere.)
    Let~$\Nu_0\colon \R \to \Rp$ denote the cumulative inflow, i.e., $\Nu_0(\vartheta) = \int_0^\vartheta \nu_0(\theta) \d\theta$ for every $\vartheta \in \R$.
    Then, $\Nu_0$ is \emph{locally absolutely continuous}, i.e., it is absolutely continuous on every bounded interval.

    \medskip
    A \emph{flow over time} in the network is defined by two functions $f^+, f^-\colon \R \to \Rp^A$ which vanish almost everywhere on $\Rnn$.
    For any $a \in A$, we denote by $f^+_a$ the function that maps $\vartheta \in \R$ to the entry $a$ of $f^+(\vartheta)$, i.e., $f^+_a(\vartheta) = \big( f^+(\vartheta) \big)_a$.
    An analogous notation is used for $f^-$ and other vector-valued functions.
    For $a \in A$, the functions $f^+_a$ and $f^-_a$ have to be locally integrable and represent the inflow and outflow rates of arcs, respectively, in dependency on the time.
    Define the cumulative flow functions $F^+, F^-\colon \R \to \Rp^A$ by
    \begin{equation*}
        F^+_a(\vartheta) := \int_0^\vartheta f^+_a(\theta) \d\theta
        \quad\text{and}\quad
        F^-_a(\vartheta) := \int_0^\vartheta f^-_a(\theta) \d\theta
        \quad\text{for all } a \in A \text{ and } \vartheta \in \R.
    \end{equation*}

    In the following, we describe the dynamics of the fluid queuing model, i.e., the conditions a flow over time has to satisfy in order to obey the model.
    Each link represented by an arc $a \in A$ has a \emph{fluid queue} at its entrance.
    The inflow and the outflow rate of $a$ are related by the dynamics of that queue, see \cref{fig:link}.
    Flow that enters the link first has to pass through the queue before it may traverse it.
    The capacity of the arc limits the rate at which flow can leave the queue.
    After leaving, the flow takes exactly $\tau_a$ units of time to traverse the link.
    We assume that queues initially are empty and \emph{operate at capacity}, i.e., as much flow as the capacity permits leaves the queue.
    Let $z_a\colon \R \to \Rp$ be the function that describes the cumulative flow which resides in the queue of $a$ in dependency on the time.
    Therefore, $z_a$ is the unique locally absolutely continuous solution (see, e.g., \citet*[p. 106]{F1988}) to
    \begin{equation*}
        \tag{QD}
        \label{eq:queuedynamics}
        z_a \equiv 0 \text{ on } \R_{\le 0}
        \quad
        \text{ and }
        \quad
        \tfrac{\text{d} z_a}{\text{d}\vartheta}(\vartheta)
        =
        \begin{cases}
            \pos[\big]{ f^+_a(\vartheta) - \nu_a } &\text{if } z_a(\vartheta) \le 0 \\
            \phantom{\big(} f^+_a(\vartheta) - \nu_a \phantom{\big)} &\text{if } z_a(\vartheta) > 0
        \end{cases}
        \quad
        \text{for a.e.\ }\vartheta \ge 0 .
    \end{equation*}
    As flow takes exactly $\tau_a$ units of time after leaving the queue to traverse $a$, we require for all times~$\vartheta \in \R$ that $z_a(\vartheta) = F^+_a(\vartheta) - F^-_a(\vartheta + \tau_a)$ holds.
    Together with \eqref{eq:queuedynamics} this determines the outflow rate of $a \in A$ as
    \begin{equation*}
        f^-_a(\vartheta + \tau_a) = f^+_a(\vartheta) - \tfrac{\text{d} z_a}{\text{d} \vartheta} (\vartheta)
        =
        \begin{cases}
            \min \set[\big]{ f^+_a(\vartheta), \nu_a } &\text{if } z_a(\vartheta) = 0 \\
            \nu_a &\text{if } z_a(\vartheta) > 0
        \end{cases} \quad\text{for a.e.\ } \vartheta \in \R.
    \end{equation*}
    Finally, we impose strict flow conservation at any vertex  $v \in V \setminus \set{t}$ but the sink, i.e.,
    \begin{equation*}
        \sum_{\mathclap{a \in \delta^+(v)}} f^+_a(\vartheta) - \sum_{\mathclap{a \in \delta^-(v)}} f^-_a(\vartheta)
        =
        \begin{cases}
            \nu_0(\vartheta) &\text{if } v = s \\
            0 &\text{if } v \in V \setminus \set{s, t}
        \end{cases}
        \quad
        \text{for a.e.\ } \vartheta \in \R .
    \end{equation*}
    Any flow over time which obeys the queuing dynamics and strict flow conservation is called \emph{feasible} in the following.

    \begin{figure}[t]
        \centering
        \begin{tikzpicture}[x=16pt,y=4pt]
            \fill[lightgray] (3,1) rectangle (5,-1);
            \shade[left color=lightgray, right color=white] (4.5,1) rectangle (5,-1);
            \shade[left color=white, right color=lightgray!50!] (6,1) rectangle (6.75,-1);
            \shade[left color=lightgray!50!, right color=white] (6.75,1) rectangle (8,-1);
            \fill[lightgray!75!] (10,1) rectangle (12,-1);
            \shade[left color=white, right color=lightgray!75!] (9,1) rectangle (11,-1);
            \fill[lightgray] (0,3) rectangle (4,-3) node[black,pos=.5] {$z_a(\vartheta)$};
            \draw (0, 3) -- (4, 3) -- (4, 1) -- (12, 1);
            \draw (0,-3) -- (4,-3) -- (4,-1) -- (12,-1);

            \draw[->,>=stealth] (0,-4) -- node[below]{\small$f_a^+(\vartheta)$} (1,-4);
            \draw[->,>=stealth] (11,-4) -- node[below]{\small$f_a^-(\vartheta)$} (12,-4);

            \draw [decorate,decoration={brace,amplitude=2pt,raise=4pt}] (0,3) -- ( 4,3) node[black,midway,yshift=12pt] {\small$q_a(\vartheta)$};
            \draw [decorate,decoration={brace,amplitude=2pt,raise=4pt}] (4.05,3) -- (12,3) node[black,midway,yshift=12pt] {\small$\tau_a$};
            \draw [decorate,decoration={brace,amplitude=2pt,raise=2pt}] (12,1) -- (12,-1) node[black,midway,xshift=12pt] {\small$\nu_a$};

        \end{tikzpicture}
        \caption{A link represented by $a \in A$ with transit time $\tau_a$, capacity $\nu_a$, inflow rate $f^+_a$, outflow rate $f^-_a$, queue length $z_a$, and queuing delay $q_a$.}
        \label{fig:link}
    \end{figure}
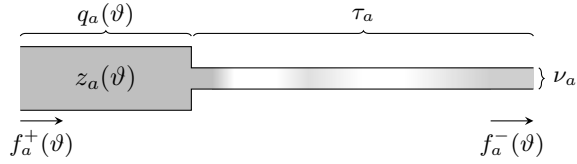

    The above defines the model from a cumulative point of view.
    Whilst it describes feasible behavior of flow in its entirety, it does not keep track of single flow particles.
    In the routing game that we are about to consider, each infinitesimally small flow particle is considered a player and, therefore, must be trackable.
    In order to allow that, the first-in-first-out principle is imposed on the queues.
    Doing so, the \emph{queuing delay} on an arc $a \in A$ which is experienced by flow entering it at time $\vartheta \in \R$ is defined as $q_a(\vartheta) := \min \set{q \ge 0\colon F^-_a(\vartheta + \tau_a + q) \ge F^+_a(\vartheta) }$.
    For queues operating at capacity, this evaluates to $q_a(\vartheta) = \tfrac{z_a(\vartheta)}{\nu_a}$.
    Therefore, a flow particle which enters an arc $a \in A$ at time~$\vartheta$ leaves it exactly at time~$T_a(\vartheta) := \vartheta + q_a(\vartheta) + \tau_a$.
    Based on strict flow conservation, flow that enters a path $P = (a_1, \ldots, a_{\card{P}}) \in A^{\card{P}}$ at time $\vartheta$ reaches its end exactly at time $\ell^P (\vartheta) := T_{a_{\card{P}}} \circ T_{a_{\card{P} - 1}} \circ \ldots \circ T_{a_1} (\vartheta)$.
    Let $\mathcal{P}_{sv}$ denote the set of all $s$--$v$-paths in~$G$.
    The \emph{earliest time} that a particle, which leaves the source at time $\vartheta \in \R$, can reach $v$ is $\ell_v(\vartheta) := \min_{P \in \mathcal{P}_{sv}} \ell^P(\vartheta)$.
    Every path attaining the minimum is called a \emph{dynamic shortest $s$--$v$-path} relative to time $\vartheta$.

    Switching from this path-based definition to an arc-based view, $\ell\colon \R \to \Rp^V$ is determined by the  Bellman equations
    \begin{equation*}
        \tag{BE}
        \label{eq:bellmanequation}
        \ell_s(\vartheta) = \vartheta
        \qquad \text{ and } \qquad
        \ell_w(\vartheta) = \min_{\mathclap{\substack{a \in \delta^-(w) \\ a = (v, w)}}} T_a \big( \ell_v(\vartheta) \big) \quad \text{ for all } w \in V \setminus \set{s} \text{ and } \vartheta \in \R .
    \end{equation*}
    The arcs for which the minima in \eqref{eq:bellmanequation} are attained are called \emph{active} relative to time $\vartheta$.
    An $s$--$v$-path in $G$ that uses only active arcs relative to $\vartheta$ is also a dynamic shortest path relative to $\vartheta$.
    Another type of arcs plays an important role in the model.
    An arc $(v, w) \in A$ with non-zero queue length at time $\ell_v(\vartheta)$ is called \emph{resetting} relative to time $\vartheta$.
    Note that $q, T$, and $\ell$ depend on the flow over time $(f^+, f^-)$.

    \section{Dynamic equilibria}
    \label{section:equilibrium}

    Interpreting each infinitesimally small flow particle in an $s$--$t$-flow over time as a player who wants to minimize her travel time from $s$ to $t$ defines a routing game over time.
    A dynamic equilibrium is a flow over time in which no player can strictly decrease her travel time by deviating on her own.
    A flow particle that starts at $s$ at time $\vartheta \in \R$ only uses active arcs relative to $\vartheta$ and, therefore, reaches every vertex $v$ on its path from $s$ to $t$ as early as possible at time~$\ell_v(\vartheta)$.
    That way, no particle starting at $s$ after time $\vartheta$ can overtake and delay it.

    \begin{definition}[Dynamic equilibrium; cf.\ \cite{CCL2015}]
        \label{definition:dynamicequilibrium}
        A feasible flow over time $(f^+, f^-)$ is a \emph{dynamic equilibrium} if for each arc $a = (v, w) \in A$ the inflow rate $f^+_a$ vanishes almost everywhere on the set~$\set{ \ell_v(\vartheta)\colon a \text{ inactive relative to } \vartheta \in \R }$.
    \end{definition}

    \citet{CCL2015} show that dynamic equilibria can be characterized in terms of cumulative flows and the earliest times function as stated in the next lemma.
    \begin{lemma}[Characterization of dynamic equilibria]
        \label{lem:dynamicequilibrium:characterization}
        A feasible flow over time $(f^+, f^-)$ is a dynamic equilibrium if and only if
        $F^+_a\big(\ell_v(\vartheta)\big) = F^-_a\big(\ell_w(\vartheta)\big)$ for all $a = (v, w) \in A$ and $\vartheta \in \R$.
    \end{lemma}

    For a flow over time with earliest times function $\ell$ and a time $\vartheta \in \R$, we define the two sets of arcs
    $A^\prime_{\vartheta} := \set[\big]{ a = (v, w) \in A\colon \ell_w(\vartheta) \ge \ell_v(\vartheta) + \tau_a }$
    and
    $A^*_{\vartheta} := \set[\big]{ a = (v, w) \in A\colon \ell_w(\vartheta) > \ell_v(\vartheta) + \tau_a }$.
    The following descriptions of these arc sets by \citet{CCL2015} give more intuition of their meaning.
    If $\ell$ is the earliest times function of a dynamic equilibrium, $A^\prime_{\vartheta}$ and $A^*_{\vartheta}$ are exactly the sets of active arcs $\set[\big]{a = (v, w) \in A\colon \ell_w(\vartheta) = T_a \big( \ell_v(\vartheta) \big) }$ relative to $\vartheta$ and resetting arcs $\set[\big]{a = (v, w) \in A\colon z_a\big( \ell_v(\vartheta) \big) > 0}$ relative to $\vartheta$, respectively.
    $A^\prime_{\vartheta}$ and $A^*_{\vartheta}$ are typically defined by this characterization.
    We use the definition based on $\ell$ only as we will need these arc sets for flows over time which are not in equilibrium and, hence, for which the characterization does not hold.

    $A^\prime_\vartheta$ and $A^*_\vartheta$ are acyclic, as $G$ does not contain cycles with total transit time zero.
    For dynamic equilibria, we additionally know that $A^\prime_\vartheta$ contains an $s$--$v$-path for every $v \in V$.
    It follows from the above characterization of $A^\prime_\vartheta$ in a dynamic equilibrium and the Bellman equations \eqref{eq:bellmanequation}.
    This makes the triple $G^\prime_\vartheta := (V, A^\prime_\vartheta, A^*_\vartheta)$ a shortest path graph with resetting in the following sense.
    \begin{definition}[Shortest path graphs with resetting]
        \label{def:shortestpathgraphswithresetting}
        A \emph{shortest path graph with resetting} is a triple $(V, A^\prime, A^*)$ such that $A^* \subseteq A^\prime \subseteq V \times V$ holds, $A^\prime$ is acyclic, and $A^\prime$ contains an $s$--$v$-path for all $v \in V$.
    \end{definition}

    \section{Normalized thin flows with resetting}
    \label{section:normalizedthinflow}

    This section repeats and extends some properties of the derivatives of dynamic equilibria.
    Assume that $(f^+, f^-)$ is a dynamic equilibrium.
    We define $x\colon \R \to \Rp^A$ by $x_a \equiv F^+_a \circ \ell_v$ for all $a = (v, w) \in A$.
    From strict flow conservation and \cref{lem:dynamicequilibrium:characterization}, it follows that for every $\vartheta \in \R$, the vector $x(\vartheta)$ is a static $s$--$t$-flow in $G$ as it satisfies
    \begin{equation*}
        \sum_{\mathclap{a \in \delta^+(v)}} x_a(\vartheta) - \sum_{\mathclap{a \in \delta^-(v)}} x_a(\vartheta)
         = \sum_{\mathclap{a \in \delta^+(v)}} F^+_a \big(\ell_v(\vartheta)\big) - \sum_{\mathclap{a \in \delta^-(v)}} F^-_a \big(\ell_v(\vartheta)\big)
         = \begin{cases}
            \Nu_0(\vartheta) &\text{if } v = s \\
            0 &\text{if } v \in V \setminus \set{s, t} .
        \end{cases}
    \end{equation*}
    The functions $x$ and $\ell$ are locally absolutely continuous.
    Therefore, the families $x$ and $\ell$ are defined by initial conditions given by the empty network and their derivatives, which exist almost everywhere.
    The derivatives fulfill the following definition, which agrees with the definition of normalized thin flows with resetting by \citet{CCL2015} with the only difference that we do not fix the label of $s$ to 1.
    \begin{definition}[Normalized thin flow with resetting]
        \label{definition:normalizedthinflow}
        For a shortest path graph with resetting $G^\prime = (V, A^\prime, A^*)$,
        a static $s$--$t$-flow $x^\prime \in \Rp^{A^\prime}$ in $(V, A^\prime)$ is called a \emph{normalized thin $s$--$t$-flow with resetting} in $G^\prime$ if there exist labels $\ell^\prime \in \Rp^V$ such that
        \begin{alignat*}{2}
            \ell^\prime_w &=  \min_{\substack{a \in \delta^-(w) \cap A^\prime \\ a = (v, w)}}\varrho^a\big( \ell^\prime_v, x^\prime_a \big) &\quad&\text{for all } w \in V \setminus \set{s}\text{, and}
            \\
            \ell^\prime_w &= \varrho^a\big( \ell^\prime_v, x^\prime_a \big) &\quad&\text{for all } a = (v, w) \in A^\prime \text{ with } x^\prime_a > 0,
        \end{alignat*}%
        where the behavior of the labels along an arc $a = (v, w) \in A^\prime$ is prescribed by the function
        \begin{equation*}
            \varrho^a\colon \Rp \times \Rp \to \Rp, \quad (\ell^\prime_v, x^\prime_a) \mapsto
            \begin{cases}
                \max \set[\big]{ \ell^\prime_v, \tfrac{x^\prime_a}{\nu_a} } &\text{if } a \in A^\prime \setminus A^*\\
                \tfrac{x^\prime_a}{\nu_a} &\text{if } a \in A^* .
            \end{cases}
        \end{equation*}
    \end{definition}

    Set $x^\prime := \tfrac{\text{d} x}{\text{d} \vartheta^+}$ and $\ell^\prime := \tfrac{\text{d} \ell}{\text{d} \vartheta^+}$ to be the right-derivatives of $x$ and $\ell$ wherever they exist, which is almost everywhere.
    Differentiating the flow conservation constraints for $x(\vartheta)$ yields flow conservation for $x^\prime(\vartheta)$.
    The Bellman equations and the conditions on a dynamic equilibrium imply that the flow $x^\prime(\vartheta)$ is a normalized thin $s$--$t$-flow of value $\nu_0(\vartheta)$ with resetting in $G^\prime_\vartheta$ and has corresponding labels $\ell^\prime(\vartheta)$.
    This characterization of the derivatives was found by~\citet{KS2011} for a slightly different notion of equilibria and investigated further by \citet{CCL2015}.

    The complementarity conditions in \cref{definition:normalizedthinflow} make complementarity problems a natural candidate for describing normalized thin flows with resetting.
    \citet{CCL2011} take this approach with a non-linear complementarity problem.
    The non-linearities in it are caused by taking minima and maxima.
    We observe that those can be resolved to linear complementarity conditions by introducing auxiliary variables.
    In order to write down the resulting linear complementarity problem in matrix form, we will need the following algebraic representation of graphs.
    \begin{definition}[Incidence matrix]
        Let $G^\prime = (V, A^\prime)$ be a directed graph.
        The (directed) incidence matrix $B \in \set{-1, 0, 1}^{V \times A^\prime}$ of $G^\prime$ is defined by
        \begin{equation*}
            B_{v, a} = \begin{cases}
                -1 & \text{if } a \in \delta^+(v) \\
                1 & \text{if } a \in \delta^-(v) \\
                0 & \text{otherwise}
            \end{cases}
            \qquad
            \text{for all } v \in V, a \in A^\prime .
        \end{equation*}
        Its positive and negative part are denoted by $B^+$ and $B^-$, respectively.
        In particular, $B = B^+ - B^-$.
    \end{definition}

    \begin{theorem}[Linear complementarity problem]
        \label{theorem:normalizedthinflow:linearcomplementarityproblem}
        Let $G^\prime = (V, A^\prime, A^*)$ be a shortest path graph with resetting and $\nu \in \Rpp^{A^\prime}$ be capacities.
        Further, let $\nu_0^\prime \ge 0$ be an inflow rate and $\ell^\prime_0 \ge 0$ a label.
        With the incidence matrix $B$ of $G^\prime$ and the diagonal matrix $D := \diag(\nu)$, define
        \begin{equation*}
            M :=
            \begin{pmatrix}
                \one_s^\top & 0 & 0 \\
                0 & B_{V \setminus \set{s}} & 0 \\
                - \big(B^+\big)^{\!\top} & \inv{D} & \Id_{\boldcdot, A^\prime \setminus A^*} \\
                - \big(B^-_{\boldcdot, A^\prime \setminus A^*}\big)^{\!\top} & \inv{D}_{A^\prime \setminus A^*, \boldcdot} & \Id
            \end{pmatrix}
            \in \R^{\big(V \cupdot A^\prime \cupdot (A^\prime \setminus A^*)\big) \times \big(V \cupdot A^\prime \cupdot (A^\prime \setminus A^*)\big)} .
        \end{equation*}
        \clearpage
        \noindent
        For $\ell^\prime \in \R^V$, $x^\prime \in \R^{A^\prime}$, and $y^\prime \in \R^{A^\prime \setminus A^*}$, the following two statements are equivalent.
        \begin{enumerate}[label=\roman*),ref=(\roman*)]
          \item \label{theorem:normalizedthinflow:linearcomplementarityproblem:lcp}
          $z^\prime := (\ell^\prime, x^\prime, y^\prime)$ is a solution to the linear complementarity problem
          \begin{equation}
            z \ge 0,
            \qquad
            M z - \ell^\prime_0 \one_s - \nu^\prime_0\one_t \ge 0,
            \qquad
            z^\top \big( M z - \ell^\prime_0 \one_s - \nu^\prime_0\one_t \big) = 0
            \label{eq:linearcomplementarityproblem}
            \tag{LCP}
          \end{equation}
          and additionally satisfies the normalization constraint $\ell^\prime_w \ge \min_{\substack{v \in N^-(w)}} \ell^\prime_v$ for every $w \in V \setminus \set{s}$ with $\delta^-(w) \cap A^* = \emptyset$ .
          \item \label{theorem:normalizedthinflow:linearcomplementarityproblem:ntf}
          $\ell^\prime$ are the corresponding labels of the normalized thin $s$--$t$-flow $x^\prime$ with resetting in $G^\prime$ of value $\nu_0^\prime$ and label $\ell^\prime_s = \ell^\prime_0$, and $y^\prime_a = \pos[\big]{\ell^\prime_v - \tfrac{x^\prime_a}{\nu_a}}$ holds for all $a = (v, w) \in A^\prime \setminus A^*$.
        \end{enumerate}
    \end{theorem}

    For the index sets of the rows and columns of the matrix $M$, we use the slight abuse of notation $V \cupdot A^\prime \cupdot (A^\prime \setminus A^*)$.
    The disjoint union of $A^\prime$ and $A^\prime \setminus A^*$ refers to the fact that we regard $A^\prime \setminus A^*$ as a copy of the respective subset of $A^\prime$ in this context.
    To indicate what row or column we refer to by an arc $a$, we explicitly distinguish between $a \in A^\prime$ and $a \in A^\prime \setminus A^*$.
    The former is associated with the variable $x^\prime_a$, while the latter corresponds to $y^\prime_a$.

    \begin{proof}
        Let $z^\prime = (\ell^\prime, x^\prime, y^\prime)$ be a solution to \eqref{eq:linearcomplementarityproblem} which satisfies $\ell^\prime_w \ge \min_{\substack{v \in N^-(w)}} \ell^\prime_v$ for every $w \in V \setminus \set{s}$ with $\delta^-(w) \cap A^* = \emptyset$.
        The two complementary inequalities for $\ell^\prime_s$ are $\ell^\prime_s \ge 0$ and $\ell^\prime_s \ge \ell^\prime_0$.
        As $\ell^\prime_0 \ge 0$, the latter has to be fulfilled with equality.
        For $a \in A^\prime \setminus A^*$, the two inequalities associated with $y^\prime_a$ are $y^\prime_a \ge 0$ and $y^\prime_a \ge \ell^\prime_v - \tfrac{x^\prime_a}{\nu_a}$.
        Together with the complementarity condition this yields $y^\prime_a = \pos[\big]{\ell^\prime_v - \tfrac{x^\prime_a}{\nu_a}}$.
        \\
        Let $w \in V$.
        For $a = (v, w) \in \delta^-(w)$, the inequality corresponding to $x^\prime_a$ yields $\ell^\prime_w \le \tfrac{x^\prime_a}{\nu_a} = \varrho^a\big(\ell^\prime_v, x^\prime_a\big)$ if $a \in A^*$ and $\ell^\prime_w \le \tfrac{x^\prime_a}{\nu_a} + y^\prime_a = \max \set[\big]{\ell^\prime_v, \tfrac{x^\prime_a}{\nu_a}} = \varrho^a\big(\ell^\prime_v, x^\prime_a\big)$ if $a \not\in A^*$.
        In both cases, the inequality is guaranteed to be tight if $x^\prime_a > 0$ due to the complementarity conditions.
        It follow that $\ell^\prime_w \le \min_{a = (v, w) \in \delta^-(w)} \varrho^a\big(\ell^\prime_v, x^\prime_a\big)$ is valid and holds with equality in case that $\delta^-(w) \cap A^* \ne \emptyset$ or $\sum_{a \in \delta^-(w)} x^\prime_a > 0$.
        Thanks to the normalization constraints, we also obtain equality in the remaining case which is $\delta^-(w) \cap A^* = \emptyset$ and $\sum_{a \in \delta^-(w)} x^\prime_a = 0$.
        \\
        We are left to show strict flow conservation for $w \in V \setminus \set{s}$.
        If $\ell^\prime_w > 0$, it is implied by the complementarity condition for $\ell^\prime_w$.
        Thus, we assume $\ell^\prime_w = 0$.
        For every $a = (v, w) \in \delta^-(w)$, we know from above that $x_a > 0$ would imply $0 = \ell^\prime_w = \varrho^a\big(\ell^\prime_v, x^\prime_a\big)$ and, hence, the contradiction $x_a = 0$.
        Assuming $w = t$ and $\nu^\prime_0 > 0$, the inequality corresponding to $\ell^\prime_w$ yields the contradiction $0 \le \sum_{a \in \delta^+(t)} x^\prime_a \le \sum_{a \in \delta^-(t)} x^\prime_a - \nu^\prime_0 < 0$.
        Otherwise, the inequality corresponding to $\ell^\prime_w$ yields $\sum_{a \in \delta^+(w)} x^\prime_a \le \sum_{a \in \delta^-(w)} x^\prime_a = 0$.
        It follows that $x^\prime_a = 0$ for all $a \in \delta(w)$ which fulfills strict flow conservation at $w$.
        \\
        In total, $\ell^\prime$ proves $x^\prime$ to be a normalized thin $s$--$t$-flow with resetting in $G^\prime$ of value $\nu^\prime_0$ with label $\ell^\prime_s = \ell^\prime_0$.

        For the converse direction, let $x^\prime$ be a normalized thin $s$--$t$-flow with resetting in $G^\prime$ of value $\nu^\prime_0$ and corresponding labels $\ell^\prime$ with $\ell^\prime_s = \ell^\prime_0$.
        Set $y^\prime_a := \pos[\big]{\ell^\prime_v - \tfrac{x^\prime_a}{\nu_a}}$ for all $a = (v, w) \in A^\prime \setminus A^*$.
        We will show that $z^\prime := (\ell^\prime, x^\prime, y^\prime)$ is a solution to \eqref{eq:linearcomplementarityproblem}.
        $z^\prime$ clearly is non-negative.
        The complementarity condition for $\ell^\prime_s$ is fulfilled, as $M_{s,\bullet} z^\prime = \ell^\prime_s = \ell^\prime_0$.
        Also, the complementarity conditions for $V \setminus \set{s}$ are met since $x^\prime$ is an $s$--$t$-flow of value $\nu^\prime_0$ and, hence, $M_{V \setminus \set{s},\bullet} z^\prime = B_{V \setminus \set{s}} x^\prime = \nu^\prime_0 \one_t$.
        \\
        Let $a = (v, w) \in A^\prime$.
        The inequality associated with $x^\prime_a$ reads $\ell^\prime_w \le \tfrac{x^\prime_a}{\nu_a} = \varrho^a\big(\ell^\prime_v, x^\prime_a\big)$ if $a \in A^*$ and $\ell^\prime_w \le \tfrac{x^\prime_a}{\nu_a} + y^\prime_a = \max \set[\big]{\ell^\prime_v, \tfrac{x^\prime_a}{\nu_a}} = \varrho^a\big(\ell^\prime_v, x^\prime_a\big)$ if $a \not\in A^*$.
        In both cases, the respective inequality is valid due to the definition of normalized thin flows with resetting.
        Further, equality is guaranteed if $x^\prime_a > 0$.
        This shows that the complementarity conditions are satisfied for the variable $x^\prime_a$.
        \\
        For $a = (v, w) \in A^\prime \setminus A^*$, the inequality corresponding to variable $y^\prime_a$ reads $y^\prime_a \ge \ell^\prime_v - \tfrac{x^\prime_a}{\nu_a}$.
        This is clearly fulfilled by the above choice of $y^\prime$.
        Further, $y^\prime_a > 0$ implies $y^\prime_a = \pos[\big]{\ell^\prime_v - \tfrac{x^\prime_a}{\nu_a}} = \ell^\prime_v - \tfrac{x^\prime_a}{\nu_a}$.
        Therefore, the complementarity condition for $y^\prime_a$ is satisfied as well.
        \\
        The normalization constraints in \ref{theorem:normalizedthinflow:linearcomplementarityproblem:lcp} hold due to the definition of normalized thin flows with resetting.
        We conclude that \ref{theorem:normalizedthinflow:linearcomplementarityproblem:ntf} implies \ref{theorem:normalizedthinflow:linearcomplementarityproblem:lcp}.
    \end{proof}

    \cref{theorem:normalizedthinflow:linearcomplementarityproblem} shows that a solution $(\ell^\prime, x^\prime, y^\prime)$ to \eqref{eq:linearcomplementarityproblem} is nearly a normalized thin flows with resettings.
    Consulting its proof reveals that the normalization constraints in \ref{theorem:normalizedthinflow:linearcomplementarityproblem:lcp} are only needed for $w \in V \setminus \set{s}$ such that $\delta^-(w) \cap A^* = \emptyset$ and $\sum_{a \in \delta^-(w)} x^\prime_a = 0$.
    In this case, \eqref{eq:linearcomplementarityproblem} allows all values $0 \le \ell^\prime_w \le \min_{v \in N^-(w)} \ell^\prime_v$ while $\ell^\prime_w = \min_{v \in N^-(w)} \ell^\prime_v$ is necessary for $\ell^\prime$ to be the corresponding labels of $x^\prime$.
    As similarly observed by \citet{CCL2011} for their non-linear complementarity problem, any solution to \eqref{eq:linearcomplementarityproblem} can be normalized to fulfill \ref{theorem:normalizedthinflow:linearcomplementarityproblem:lcp}.
    For that purpose, we define a normalization function $\pi\colon \R^V \to \R^V$.
    Let $\ell^\prime$ be the labels of a solution to \eqref{eq:linearcomplementarityproblem}.
    Fix an arbitrary topological order of the vertices in $G^\prime$ and successively set for $w \in V$ in that order
    \begin{equation*}
      \pi_w(\ell^\prime)
      =
      \begin{cases}
          \max\set{\ell^\prime_w, \min_{v \in N^-(w)} \pi_v(\ell^\prime)}
          &\text{if } w \ne s \text{ and } \delta^-(w) \cap A^* = \emptyset
      \\
          \ell^\prime_w
          &\text{otherwise}
      \end{cases} .
    \end{equation*}
    Note that the definition is independent of the particular topological order that was chosen.

    \begin{lemma}
        \label{lemma:normalizedthinflow:normalization}
        The normalization function $\pi\colon \R^V \to \R^V$ as defined above satisfies
        \begin{enumerate}[label=\roman*),ref=(\roman*)]
            \item \label{lemma:normalizedthinflow:normalization:monotone}
            $\ell^\prime \le \pi(\ell^\prime)$ for all $\ell^\prime \in \Rp^V$, and
            \item \label{lemma:normalizedthinflow:normalization:lipschitz}
            $\norm{\pi(\ell^\doubleprime) - \pi(\ell^{\prime})}_\infty \le \norm{\ell^\doubleprime - \ell^{\prime}}_\infty$ for all $\ell^\prime, \ell^\doubleprime \in \Rp^V$.
        \end{enumerate}
    \end{lemma}

    \begin{proof}
        \ref{lemma:normalizedthinflow:normalization:monotone} follows immediately from the definition of $\pi$.
        To see \ref{lemma:normalizedthinflow:normalization:lipschitz}, we prove by induction on $w \in V$ in topological order that $\abs{\pi_w(\ell^\doubleprime) - \pi_w(\ell^{\prime})} \le \norm{\ell^\doubleprime - \ell^{\prime}}_\infty$.
        For $w = s$ and all $w \in V$ such that $\delta^-(w) \cap A^* \ne \emptyset$, by definition $\pi_w(\ell^\prime) = \ell^\prime_w$ and, hence, the statement holds.
        For $w \in V \setminus \set{s}$ with $\delta^-(w) \cap A^* = \emptyset$, assume that the induction hypothesis holds for all vertices $v$ which are topologically preceding $w$.
        Then, the claim also holds for $w$ based on
        \begin{equation*}
            \abs{\pi_w(\ell^\doubleprime) - \pi_w(\ell^{\prime})}
            \le
            \max \set[\Big]{\abs{\ell^\doubleprime_w - \ell^{\prime}_w}, \max_{v \in N^-(w)} \abs{\pi_v(\ell^\doubleprime) - \pi_v(\ell^{\prime})}}
            \le
            \norm{\ell^\doubleprime - \ell^\prime}_\infty .
            \tag*{\qedhere}
        \end{equation*}
    \end{proof}

    \begin{corollary}[Normalization]
        \label{corollary:normalizedthinflow:linearcomplementarityproblem}
        Let $G^\prime = (V, A^\prime, A^*)$ be a shortest path graph with resetting and $\nu \in \Rpp^{A^\prime}$ be capacities.
        Further, let $\nu_0^\prime \ge 0$ be an inflow rate and $\ell^\prime_0 \ge 0$ a label.
        If $z^\prime = (\ell^\prime, x^\prime, y^\prime)$ is a solution to \eqref{eq:linearcomplementarityproblem}, then $\pi(\ell^\prime)$ are the corresponding labels of the normalized thin $s$--$t$-flow $x^\prime$ with resetting in $G^\prime$ of value $\nu^\prime_0$ with label $\ell^\prime_s = \ell^\prime_0$.
    \end{corollary}

    \begin{proof}
        Let $z^\prime = (\ell^\prime, x^\prime, y^\prime)$ be a solution to \eqref{eq:linearcomplementarityproblem}.
        Let $z^\doubleprime = (\ell^\doubleprime, x^\doubleprime, y^\doubleprime)$ be another solution to \eqref{eq:linearcomplementarityproblem} such that $x^\doubleprime = x^\prime$, $\ell^\doubleprime \le \pi(\ell^\prime)$, and the set $U := \set{v \in V\colon \ell^\doubleprime_v < \pi_v(\ell^\prime)}$ is inclusion-minimal.
        Note that $z^\doubleprime$ exists since \cref{lemma:normalizedthinflow:normalization} guarantees that $z^\prime$ is feasible to the optimization problem defining $z^\doubleprime$.
        As seen in the proof of \cref{theorem:normalizedthinflow:linearcomplementarityproblem} for $a = (v, w) \in A$, the inequalities corresponding to $x_a$ for $z^\prime$ and $z^\doubleprime$ read $\ell^\prime_w \le \varrho^a\big(\ell^\prime_v, x^\prime_a\big)$ and $\ell^\doubleprime_w \le \varrho^a\big(\ell^\doubleprime_v, x^\doubleprime_a\big)$, respectively.
        Both hold with equality if $a \in A^*$ or $x^\prime_a = x^\doubleprime_a > 0$.
        \\
        Assume $U \ne \emptyset$.
        From the definition of $\pi$, we obtain $s \not\in U$.
        Let $w \in U$ be topologically minimal, in particular $N^-(w) \cap U = \emptyset$.
        Assume there is $a = (v, w) \in \delta^-(w)$ such that $a \in A^*$ or $x^\prime_a > 0$.
        Then $\ell^\prime_w = \varrho^a\big( \ell^\prime_v, x^\prime_a \big) \le \varrho^a\big( \pi_v(\ell^\prime), x^\prime_a \big) = \varrho^a\big( \ell^\doubleprime_v, x^\doubleprime_a \big) = \ell^\doubleprime_w < \pi_w(\ell^\prime)$ shows $a \not\in A^*$.
        But then $\ell^\doubleprime_w < \pi_w(\ell^\prime) = \min_{u \in N^-(w)} \pi_u(\ell^\prime) = \min_{u \in N^-(w)} \ell^\doubleprime_u \le \ell^\doubleprime_v = \varrho^a\big(\ell^\doubleprime_v, x^\doubleprime_a \big)$ contradicts the complementarity requirements.
        \\
        We conclude that $\delta^-(w) \cap A^* = \emptyset$ and $x^\prime$ vanishes on $\delta^-(w)$.
        Since weak flow conservation is part of \eqref{eq:linearcomplementarityproblem}, $x^\prime$ vanishes on $\delta^+(w)$ as well.
        Therefore, increasing $\ell^\doubleprime_w$ and $y^\doubleprime_a$ for all $a \in \delta^+(w)$ to $\min_{v \in N^-(w)} \ell^\doubleprime_v$ yields a solution to \eqref{eq:linearcomplementarityproblem} with strictly smaller $U$ and contradicts the choice of $z^\doubleprime$.
        Hence, $U = \emptyset$ and $\ell^\doubleprime = \pi(\ell^\prime)$ have to hold.
        Consequently, $z^\doubleprime = \big(\pi(\ell^\prime), x^\prime, y^\doubleprime\big)$ also fulfills the normalization constraints in  \cref{theorem:normalizedthinflow:linearcomplementarityproblem} \ref{theorem:normalizedthinflow:linearcomplementarityproblem:lcp} which implies the statement.
    \end{proof}

    \cref{corollary:normalizedthinflow:linearcomplementarityproblem} allows us to use theory on linear complementarity problems to study normalized thin flows with resetting.
    In the literature, results on linear complementarity problems are typically stated in terms of the properties of their matrices.
    The next theorem regards the sign of the principal minors of the matrix $M$ from \eqref{eq:linearcomplementarityproblem}.
    We need to clarify what is meant when speaking about the determinant of a square matrix without a linear ordering of its index set.
    Let $I$ be a finite index set of size $k := \card{I}$ and $K \in \R^{I \times I}$ be a square matrix.
    For an arbitrary permutation $\sigma\colon [k] \to I$, define $\det{K} := \det*{\left(K_{\sigma(i), \sigma(j)}\right)_{i, j \in [k]}}$.
    Note that $\det{K}$ is independent of $\sigma$ as the determinant is invariant under symmetric permutation of the rows and columns.

    \begin{lemma}[Principal minors of $M$]
        \label{lemma:linearcomplementarityproblem:p0}
        Let $G^\prime = (V, A^\prime, A^*)$ be a shortest path graph with resetting and $\nu \in \Rpp^{A^\prime}$ be capacities.
        Further, let $M$ be the corresponding matrix of \eqref{eq:linearcomplementarityproblem}.
        Then, all principal minors of $M$ are non-negative.
    \end{lemma}

    For the prove of this lemma, we will need (another) algebraic representation of graphs, which is the Laplacian matrix.
    It naturally appears due to its relation to the incidence matrix, which we describe next.
    The principal minors of $M$ (in the algebraic sense) are closely related to the determinants of the Laplacian matrix of the minors of $G^\prime$ (in the graph theoretic sense).
    In that way, the non-negativity of the minors of $M$ are ultimately based on the same property of the minors of Laplacian matrices.

    \begin{definition}[Weighted Laplacian matrix]
        \label{definition:laplacian}
        Let $G = (V, A^\prime)$ be a directed graph with arc weights $\nu \in \Rpp^{A^\prime}$.
        The weighted (directed) Laplacian matrix $L \in \R^{V \times V}$ of $G$ is defined by
        \begin{equation*}
            L_{v, w}
            =
            \begin{cases}
                \sum_{a \in \delta^-(v)} \nu_a
                & \text{if } v = w
            \\
                -\nu_{v, w}
                & \text{if } (v, w) \in A^\prime
            \\
                0
                & \text{otherwise}
            \end{cases}
            \qquad
            \text{for all } v, w \in V.
        \end{equation*}
    \end{definition}

    \begin{lemma}[Incidence and Laplacian matrix]
        \label{lemma:laplacian:incidence}
        Let $G = (V, A^\prime)$ be a directed graph with arc weights $\nu \in \Rpp^{A^\prime}$.
        Then, its incidence matrix $B \in \set{-1, 0, 1}^{V \times A^\prime}$ and its weighted Laplacian matrix $L \in \R^{V \times V}$ satisfy $L = B D \big(B^+\big)^{\!\top}$ where $D = \diag(\nu)$.
    \end{lemma}

    \begin{proof}
        For every $v, w \in V$, the definitions of the incidence matrix and the Laplacian matrix yield
        \begin{equation}
            B_{v, \boldcdot} D \big( B^+_{w, \boldcdot} \big)^{\!\top}
            =
            \big( \one_{\delta^-(v)} - \one_{\delta^+(v)} \big)^{\!\top} D \one_{\delta^-(w)}
            =
            \left\{
            \begin{aligned}
                \one_{\delta^-(v)}^\top D \one_{\delta^-(v)} \enskip & \text{if } v = w \\
                -\one_{\delta^+(v)}^\top D \one_{\delta^-(w)} \enskip & \text{if } v \ne w
            \end{aligned}
            \right\}
            = L_{v, w} .
            \tag*{\qedhere}
        \end{equation}
    \end{proof}

    \begin{lemma}[Principal minors of a Laplacian]
        \label{lemma:laplacian:p0}
        All principal minors of a graph's weighted Laplacian matrix are non-negative.
    \end{lemma}

    \begin{proof}
        Let $L \in \R^{V \times V}$ be the weighted Laplacian of a directed graph $G = (V, A^\prime)$ with weights $\nu \in \Rpp^{A^\prime}$.
        For $U \subseteq V$, the principal submatrix $L_{U, U}$ is (weakly) column diagonally dominant, as for every $w \in U$
        \begin{equation*}
            L_{w, w} = \enskip\sum_{\mathclap{a \in \delta^-(w)}} \nu_a \ge \quad\sum_{\mathclap{v \in N^-(w) \cap U}} \nu_{v, w} = \enskip\sum_{\mathclap{v \in U \setminus \set{w}}} \abs{L_{v, w}} .
        \end{equation*}
        It follows by the Ger\v{s}gorin circle theorem (see~\cite{G1931}) that all real eigenvalues of $L_{U, U}$ are non-negative and, hence, $\det[\big]{L_{U, U}} \ge 0$.
    \end{proof}

    \begin{proof}[Proof of \cref{lemma:linearcomplementarityproblem:p0}.]
        Let an arbitrary principal submatrix of $M$ be given by the index sets $U \subseteq V$, $X \subseteq A^\prime$, and $Y \subseteq A^\prime \setminus A^*$.
        We need to show $\det[\big]{M_{U \cupdot X \cupdot Y, U \cupdot X \cupdot Y}} \ge 0$.
        If $s \in U$, a Laplacian expansion along the row of $s$ shows $\det[\big]{M_{U \cupdot X \cupdot Y, U \cupdot X \cupdot Y}} = \det[\big]{M_{U \setminus \set{s} \cupdot X \cupdot Y, U \setminus \set{s} \cupdot X \cupdot Y}}$.
        Hence, we can assume that $s \not\in U$.
        Refining the block structure of $M_{U \cupdot X \cupdot Y, U \cupdot X \cupdot Y}$ by partitioning $X$ into $X \setminus Y$ and $X \cap Y$ yields
        \begin{equation*}
            \det[\big]{M_{U \cupdot X \cupdot Y, U \cupdot X \cupdot Y}}
            =
            \det*{\begin{pmatrix}
                0 & B_{U,X \setminus Y} & B_{U,X \cap Y} & 0 \\
                -\big(B^+_{U,X \setminus Y}\big)^{\!\top} & \inv{D_{X \setminus Y, X \setminus Y}} & 0 & 0 \\
                -\big(B^+_{U,X \cap Y}\big)^{\!\top} & 0 & \inv{D_{X \cap Y, X \cap Y}} & \Id_{X \cap Y, Y} \\
                -\big(B^-_{U,Y}\big)^{\!\top} & 0 & \inv{D}_{Y,X \cap Y} & \Id_{Y, Y}
            \end{pmatrix}} .
        \end{equation*}
        We multiply the rows in $X \setminus Y$ from the left by $B_{U,X \setminus Y} D_{X \setminus Y, X \setminus Y}$ and subtract the result from the rows corresponding to $U$.
        Further, we subtract the rows in $X \cap Y \subseteq Y$ from the corresponding rows $X \cap Y \subseteq X$.
        As these are unitary row operations under which the determinant is invariant, it holds
        \begin{equation*}
            \det[\big]{M_{U \cupdot X \cupdot Y, U \cupdot X \cupdot Y}}
            =
            \det*{\begin{pmatrix}
                B_{U,X \setminus Y} D_{X \setminus Y, X \setminus Y} \big(B^+_{U,X \setminus Y}\big)^{\!\top} & 0 & B_{U,X \cap Y} & 0 \\
                -\big(B^+_{U,X \setminus Y}\big)^{\!\top} & \inv{D_{X \setminus Y, X \setminus Y}} & 0 & 0 \\
                -\big(B_{U,X \cap Y}\big)^{\!\top} & 0 & 0 & 0 \\
                -\big(B^-_{U,Y}\big)^{\!\top} & 0 & \inv{D}_{Y,X \cap Y} & \Id_{Y, Y}
            \end{pmatrix}} .
        \end{equation*}
        Symmetric rearranging of the rows and columns reveals the triangular block structure
        \begin{align*}
            \det[\big]{M_{U \cupdot X \cupdot Y, U \cupdot X \cupdot Y}}
            &=
            \det*{\begin{pmatrix}
                B_{U,X \setminus Y} D_{X \setminus Y, X \setminus Y} \big(B^+_{U,X \setminus Y}\big)^{\!\top} & B_{U,X \cap Y} & 0 & 0 \\
                -\big(B_{U,X \cap Y}\big)^{\!\top} & 0 & 0 & 0 \\
                -\big(B^-_{U,Y}\big)^{\!\top} & \inv{D}_{Y,X \cap Y} & \Id_{Y, Y} & 0 \\
                -\big(B^+_{U,X \setminus Y}\big)^{\!\top} & 0 & 0 & \inv{D_{X \setminus Y, X \setminus Y}}
            \end{pmatrix}} .
        \end{align*}
        By \cref{lemma:laplacian:incidence}, the matrix $B_{\boldcdot,X \setminus Y} D_{X \setminus Y, X \setminus Y} \big(B^+_{\boldcdot,X \setminus Y}\big)^{\!\top}$ is the weighted Laplacian matrix of the graph $H := \big(V, X \setminus Y\big)$.
        Therefore, we will denote it by $L^H$.
        Using the multiplicativity of the determinant for triangular block matrices results in
        \begin{align*}
          \det[\big]{M_{U \cupdot X \cupdot Y, U \cupdot X \cupdot Y}}
            &=
            \det{\inv{D_{X \setminus Y, X \setminus Y}}}
            \det*{\begin{pmatrix}
                L^H_{U, U} & B_{U,X \cap Y} \\
                -\big(B_{U,X \cap Y}\big)^{\!\top} & 0
            \end{pmatrix}} .
        \end{align*}

        If there exists $a \in (X \cap Y) \cap \big( (V \setminus U) \times (V \setminus U) \big)$, i.e., $a \in X \cap Y$ is not incident to any vertex in~$U$, then $B_{U,a} = 0$.
        Therefore, $M_{U \cupdot X \cupdot Y, U \cupdot X \cupdot Y}$ is singular in that case.
        On the other hand, assume there exists $a \in (X \cap Y) \cap \delta(U)$, i.e., there is exactly one $u \in U$ which is incident to $a \in X \cap Y$.
        Applying the Laplacian expansion consecutively along the row $a$ and column $a$ yields
        \begin{equation*}
            \det[\big]{M_{U \cupdot X \cupdot Y, U \cupdot X \cupdot Y}}
            =
            \det{\inv{D_{X \setminus Y, X \setminus Y}}}
            B_{u, a}^2
            \det*{\begin{pmatrix}
                L^H_{U \setminus \set{u}, U \setminus \set{u}} & B_{U \setminus \set{u}, X \cap Y} \\
                -\big(B_{U \setminus \set{u}, X \cap Y}\big)^{\!\top} & 0
            \end{pmatrix}} .
        \end{equation*}
        As $B_{u, a}^2 \ge 0$, we can assume that $X \cap Y \subseteq U \times U$ by using induction on $\card{U}$.
        Under this assumption, if $X \cap Y$ contains an undirected cycle, it follows that $B_{U, X \cap Y}$ does not have full column rank and $M_{U \cupdot X \cupdot Y, U \cupdot X \cupdot Y}$ is singular.
        Thus, we additionally assume in the following that $X \cap Y$ does not contain any undirected cycle.

        Let $U = \dot\bigcup_{i \in [k]} U_i$ and $X \cap Y = \dot\bigcup_{i \in [k]} Y_i$ give the partition of $(U, X \cap Y)$ into its $k \in \N$ weakly connected components $(U_1, Y_1), \dots, (U_k, Y_k)$.
        Then $(U_i, Y_i)$ is a tree when ignoring the orientation of the arcs for every $i \in [k]$.
        Fix an arbitrary root vertex $r_i \in U_i$ for every $i \in [k]$ and define the set $R := \set{r_1, \dots, r_k}$.
        Refining the block structure of the remaining matrix by splitting $U$ into $U \setminus R$ and $R$ gives that
        \begin{equation*}
            \det[\big]{M_{U \cupdot X \cupdot Y, U \cupdot X \cupdot Y}} = \det{\inv{D_{X \setminus Y, X \setminus Y}}}
            \det*{\begin{pmatrix}
                L^H_{U \setminus R, U \setminus R} & L^H_{U \setminus R, R} & B_{U \setminus R, X \cap Y} \\
                L^H_{R, U \setminus R} & L^H_{R, R} & B_{R, X \cap Y} \\
                -B_{U \setminus R, X \cap Y}^\top & -B_{R, X \cap Y}^\top & 0
            \end{pmatrix}} .
        \end{equation*}
        For each $i \in [k]$ and $u \in U_i \setminus \set{r_i}$, we sequentially add row $u$ to row $r_i$ and column $u$ to column~$r_i$.
        Note that the determinant is invariant under these unitary operations and the resulting matrix does not depend on their order.
        The effect of the operations on $L^H_{U, U}$ can be interpreted as the contraction of arcs in the following sense.
        Let $\widehat{H} := H / (X \cap Y)$ be the graph that results from contracting all arcs of $X \cap Y$ in $H$, and $L^{\widehat{H}}$ be its Laplacian matrix.
        For every $i \in [k]$, we identify $r_i$ with the vertex which results from contracting $U_i$ by $Y_i$.
        As $(U_i, Y_i)$ is a weakly connected component with respect to $X \cap Y$, we obtain $\one_{U_i}^\top B_{U, X \cap Y} = 0$ for all $i \in [k]$.
        Putting this together gives
        \begin{equation*}
            \det[\big]{M_{U \cupdot X \cupdot Y, U \cupdot X \cupdot Y}}
            =
            \det{\inv{D_{X \setminus Y, X \setminus Y}}}
            \det*{\begin{pmatrix}
                L^H_{U \setminus R, U \setminus R} & * & B_{U \setminus R, X \cap Y} \\
                * & L^{\widehat{H}}_{R, R} & 0 \\
                -B_{U \setminus R, X \cap Y}^\top & 0 & 0
            \end{pmatrix}} ,
        \end{equation*}
        where $*$ marks blocks the specific value of which is not relevant for our purposes.
        Since the components $(U_i, Y_i)$ are trees, the matrix $B_{U \setminus R, X \cap Y}$ is square.
        By swapping the columns $U \setminus R$ with the columns $X \cap Y$, we obtain
        \begin{equation*}
            \det[\big]{M_{U \cupdot X \cupdot Y, U \cupdot X \cupdot Y}}
            =
            \det{\inv{D_{X \setminus Y, X \setminus Y}}}
            (-1)^{\card{X \cap Y}}
            \det*{\begin{pmatrix}
                B_{U \setminus R, X \cap Y} & * & L^H_{U \setminus R, U \setminus R} \\
                0 & L^{\widehat{H}}_{R, R} & * \\
                0 & 0 & -B_{U \setminus R, X \cap Y}^\top
            \end{pmatrix}} .
        \end{equation*}
        Exploiting the triangular block structure again and applying \cref{lemma:laplacian:p0} finally yields the statement
        \begin{equation}
            \det[\big]{M_{U \cupdot X \cupdot Y, U \cupdot X \cupdot Y}}
            =
            \det{\inv{D_{X \setminus Y, X \setminus Y}}}
            \det[\big]{B_{U \setminus R, X \cap Y}}^2
            \det{L^{\widehat{H}}_{R, R}}
            \ge 0 .
        \tag*{\qedhere}
        \end{equation}
    \end{proof}

    From \cref{lemma:linearcomplementarityproblem:p0}, we obtain a proof for the existence of normalized thin $s$--$t$-flows with resetting.
    It is an alternative to the original proof by \citet{CCL2015}, which is based on an elegant application of Kakutani's fixed point theorem.

    \begin{theorem}[Existence]
        \label{thm:normalizedthinflow:existence}
        Let $G^\prime = (V, A^\prime, A^*)$ be a shortest path graph with resetting and let $\nu \in \Rpp^{A^\prime}$ be capacities.
        For all given values $\nu_0^\prime, \ell^\prime_0 \ge 0$, there exists a normalized thin $s$--$t$-flow with resetting of value $\nu_0^\prime$ in $G^\prime$ and corresponding label $\ell^\prime_0$ at $s$.
    \end{theorem}

    \begin{proof}
        Due to a result by \citet[Corollary 3.9.22]{CPS2009} and \cref{corollary:normalizedthinflow:linearcomplementarityproblem}, it is enough to show that $z^\prime = 0$ is the unique solution to \eqref{eq:linearcomplementarityproblem} for $\nu^\prime_0 = 0$ and $\ell^\prime_0 = 0$, and that its matrix $M$ has only non-negative minors.
        The latter is taken care of by \cref{lemma:linearcomplementarityproblem:p0}.
        For the flow rate $\nu^\prime_0 = 0$ and label $\ell^\prime_0 = 0$, the unique $s$--$t$-flow $x^\prime \equiv 0$ is also a normalized thin $s$--$t$-flow with resetting in $G^\prime$ with unique corresponding labels $\ell^\prime \equiv 0$.
        By \cref{theorem:normalizedthinflow:linearcomplementarityproblem,corollary:normalizedthinflow:linearcomplementarityproblem}, $z^\prime = 0$ is therefore the unique solution to \eqref{eq:linearcomplementarityproblem} for these parameters.
    \end{proof}

    The proof of \cref{thm:normalizedthinflow:existence} uses the uniqueness of the solution $z^\prime = 0$ to \eqref{eq:linearcomplementarityproblem} for $\nu^\prime_0 = 0$ and $\ell^\prime_0 = 0$.
    Together with \cref{lemma:linearcomplementarityproblem:p0}, it even proves the applicability of known pivoting methods and iterative methods for linear complementarity problems to \eqref{eq:linearcomplementarityproblem}.
    In particular, a result by \citet[Theorems 4.4.8 and 4.4.11]{CPS2009} shows that Lemke's algorithm can be used to find a normalized thin flow with resetting in finitely many steps, when dealing with degeneracy appropriately (see \cite[Section 4.9]{CPS2009}).

    We want to examine the dependency of normalized thin flows with resetting on the flow value and the label of $s$.
    The following proof for the monotonicity of the corresponding labels in these two parameters is a refinement of the analysis that \citet{CCL2015} use to show uniqueness of the labels (which is an immediate corollary of the monotonicity).

    \begin{theorem}[Monotonicity of labels]
        \label{theorem:normalizedthinflow:monotonicity}
        Let $G^\prime = (V, A^\prime, A^*)$ be a shortest path graph with resetting and $\nu \in \Rpp^{A^\prime}$ be capacities.
        Further, let $x^\prime$, $\hat{x}^\prime$ be two normalized thin $s$--$t$-flows with resetting in $G^\prime$ and let $\ell^\prime, \hat{\ell}^\prime$ be corresponding labels, respectively.
        If the flow values fulfill $\card{x^\prime} \le \card{\hat{x}^\prime}$ and the labels fulfill $\ell^\prime_s \le \hat{\ell}^\prime_s$, then $\ell^\prime \le \hat{\ell}^\prime$ holds (element-wise).
    \end{theorem}

    \begin{proof}
        In the following, all the incidences $\delta^-$, $\delta^+$, and $\delta$ relate to $(V, A^\prime)$.
        Assume for contradiction that $\ell^\prime_s \le \hat{\ell}^\prime_s$, $\card{x^\prime} \le \card{\hat{x}^\prime}$, and $U := \set[\big]{v \in V\colon \ell^\prime_v > \hat{\ell}^\prime_v} \ne \emptyset$ hold.

        We claim that $x^\prime_a = \hat{x}^\prime_a$ for all $a \in \delta(U)$.
        Assume again this is wrong.
        We know $s \not\in U$ and, thus,
        \begin{equation*}
            \sum_{\mathclap{a \in \delta^+(U)}} (x^\prime_a - \hat{x}^\prime_a) - \sum_{\mathclap{a \in \delta^-(U)}} (x^\prime_a - \hat{x}^\prime_a)
            =
            \sum_{\mathclap{u \in U}} \bigg( \Big( \sum_{\mathclap{a \in \delta^+(u)}} x^\prime_a - \sum_{\mathclap{a \in \delta^-(u)}} x^\prime_a\Big) - \Big( \sum_{\mathclap{a \in \delta^+(u)}} \hat{x}^\prime_a - \sum_{\mathclap{a \in \delta^-(u)}} \hat{x}^\prime_a\Big) \bigg)
            \ge 0 .
        \end{equation*}
        Therefore, there has to be $a = (v, w) \in A^\prime$ such that $a \in \delta^+(U)$ and $x^\prime_a > \hat{x}^\prime_a$, or $a \in \delta^-(U)$ and $x^\prime_a < \hat{x}^\prime_a$.
        If $a \in \delta^+(U)$ then $\ell^\prime_w = \varrho^a\big(\ell^\prime_v, x^\prime_a\big) > \varrho^a\big(\hat{\ell}^\prime_v, \hat{x}^\prime_a\big) \ge \hat{\ell}^\prime_w$ which contradicts $w \notin U$.
        If $a \in \delta^-(U)$ then $\hat{\ell}^\prime_w = \varrho^a\big(\hat{\ell}^\prime_v, \hat{x}^\prime_a\big) \ge \varrho^a\big(\ell^\prime_v, x^\prime_a\big) \ge \ell^\prime_w$ which contradicts $w \in U$.
        Thus, $x^\prime$ and $\hat{x}^\prime$ agree on $\delta(U)$.

        As a consequence, $\delta^-(U) \cap A^* = \emptyset$ and $\varrho^a\big(\ell^\prime_v, x^\prime_a\big) = \ell^\prime_v$ for all $a = (v, w) \in \delta^-(U)$.
        Since $A^\prime$ is acyclic and $s \not\in U$, there is $w \in U$ such that $\emptyset \ne \delta^-(w) \subseteq \delta^-(U)$.
        Then, $w \in U$ contradicts
        \begin{equation*}
            \ell^\prime_w
            =
            \enskip\min_{\mathclap{\substack{a \in \delta^-(w) \\ a = (v, w)}}} \enskip \varrho^a\big(\ell^\prime_v, x^\prime_a\big)
            =
            \enskip\min_{\mathclap{v \in N^-(w)}} \enskip \ell^\prime_v
            \le
            \enskip\min_{\mathclap{v \in N^-(w)}} \enskip \hat{\ell}^\prime_v
            \le
            \enskip\min_{\mathclap{\substack{a \in \delta^-(w) \\ a = (v, w)}}} \enskip \varrho^a\big(\hat{\ell}^\prime_v, \hat{x}^\prime_a\big)
            =
            \hat{\ell}^\prime_w . \tag*{\qedhere}
        \end{equation*}
    \end{proof}

    \begin{corollary}[Uniqueness of labels]
        \label{corollary:normalizedthinflow:uniqueness}
        Let $G^\prime = (V, A^\prime, A^*)$ be a shortest path graph with resetting and $\nu \in \Rpp^{A^\prime}$ be capacities.
        The corresponding labels $\ell^\prime$ of a normalized thin $s$--$t$-flows $x^\prime$ with resetting in $G^\prime$ are uniquely determined by its flow value $\card{x^\prime}$ and the label $\ell^\prime_s$.
    \end{corollary}

    \begin{example}[Non-uniqueness of thin flows]
        \label{example:normalizedthinflow:nonuniqueness}
        Consider the shortest path graph with resetting $G^\prime = (V, A^\prime, A^*)$ consisting of the source $s$, the sink $t$, and two parallel non-resetting arcs $a$ and $b$ from $s$ to $t$, i.e., $V = \set{s, t}$, $A^\prime = \set{a, b}$ and $A^* = \emptyset$.
        Let $\nu \equiv 1$.
        Then for every $0 \le x^\prime_a \le 1$ setting $x^\prime_{b} = 1 - x^\prime_a$ defines a normalized thin $s$--$t$-flow with resetting of value $1$ in $G^\prime$ with corresponding labels $\ell^\prime_s = \ell^\prime_t = 1$.
    \end{example}

    \cref{corollary:normalizedthinflow:uniqueness} shows that the corresponding labels $\ell^\prime$ of a normalized thin $s$--$t$-flow with resetting are uniquely determined by the label $\ell^\prime_s$ and the flow rate $\nu^\prime_0$.
    In contrast to that, \cref{example:normalizedthinflow:nonuniqueness} shows that the flow itself is not necessarily uniquely determined by those parameters.
    In some sense, however, its non-uniqueness is the only kind that can occur.
    The flow on a subset of arcs is uniquely determined by the labels.
    For every $a = (v, w) \in A^\prime \setminus A^*$ with $\ell^\prime_v < \ell^\prime_w$ and for every $a = (v, w) \in A^*$, it holds $x^\prime = \ell^\prime_w \nu_a$.
    For every $a = (v, w) \in A^\prime \setminus A^*$ with $\ell^\prime_v > \ell^\prime_w$ on the other hand, $x^\prime_a = 0$.
    Therefore, non-uniqueness can only arise in weakly connected components with respect to $A^\prime \setminus A^*$ of constant label.

    The proof of \cref{corollary:normalizedthinflow:uniqueness} is combinatorial in nature and it is not clear whether a similar result can be obtain by the means of linear complementarity problems.
    There are results on the uniqueness of the solutions to linear complementarity problems.
    Those, however, do not apply to \eqref{eq:linearcomplementarityproblem}.
    Apart from the discussed non-uniqueness of \cref{example:normalizedthinflow:nonuniqueness},  \eqref{eq:linearcomplementarityproblem} suffers from a different kind which is the non-normalized labels.
    The latter can be addressed as discussed at the end of this section.
    Yet, it would only yield uniqueness of the labels.
    To the knowledge of the author, there are no results on such partial uniqueness in the general theory on linear complementarity problems.

    The flow value $\nu^\prime_0$ appears in a linear way in \eqref{eq:linearcomplementarityproblem}.
    This allows to treat it as a variable and get a new linear complementarity problem.
    Its set of solutions captures the dependency of normalized thin flows with resetting on the flow value, which is analyzed in the following lemmas and used in \cref{section:evolution,section:seriesparallel}.

    \begin{proposition}[Parametric normalized thin flows with resetting]
        \label{lem:normalizedthinflow:parametric}
        Let $G^\prime = (V, A^\prime, A^*)$ be a shortest path graph with resetting and $\nu \in \Rpp^{A^\prime}$ be capacities.
        Then there are continuous, piecewise linear functions $\varrho^{G^\prime}\colon \Rp \to \Rp^V$ and $\chi^{G^\prime}\colon \Rp \to \Rp^{A^\prime}$ such that, for all $\nu_0^\prime \ge 0$, the vector $\chi^{G^\prime}\!(\nu_0^\prime)$ is a normalized thin $s$--$t$-flow with resetting of value $\nu_0^\prime$ in $G^\prime$ and $\varrho^{G^\prime}\!(\nu_0^\prime)$ are corresponding labels with $\varrho_s^{G^\prime}\!(\nu_0^\prime) = 1$.
    \end{proposition}
    \begin{proof}
        Define the function $\varrho^{G^\prime}\colon \Rp \to \Rp^V$ by setting, for every $\nu_0^\prime \ge 0$, the vector $\varrho^{G^\prime}\!(\nu_0^\prime)$ to be the corresponding labels of a normalized thin $s$--$t$-flow with resetting  of value $\nu_0^\prime$ in $G^\prime$ such that the corresponding label at $s$ is $1$.
        By \cref{thm:normalizedthinflow:existence,corollary:normalizedthinflow:uniqueness}, this is a sound definition.
        According to \citet*[Theorem 7.2.1]{CPS2009}, the set of solutions to a linear complementarity problem depends in a locally upper Lipschitz continuous way on its right hand side as follows.
        Fix $\nu^\prime_0 \ge 0$.
        Let $z^\doubleprime = (\ell^\doubleprime, x^\doubleprime, y^\doubleprime)$ be a solution to \eqref{eq:linearcomplementarityproblem} with $\nu^\prime_0$ replaced by another flow rate $\nu^\doubleprime_0 \ge 0$.
        There exists a constant $C > 0$ and $\varepsilon > 0$ such that $\abs{\nu^\doubleprime_0 - \nu^{\prime}_0} \le \varepsilon$ implies the existence of a solution $z^\prime = (\ell^\prime, x^\prime, y^\prime)$ to \eqref{eq:linearcomplementarityproblem} satisfying
        $\norm{z^\doubleprime - z^\prime}_\infty \le C \abs{\nu^\doubleprime_0 - \nu^{\prime}_0}$.
        Together with \cref{lemma:normalizedthinflow:normalization,corollary:normalizedthinflow:linearcomplementarityproblem}, it follows
        \begin{equation*}
            \norm[\big]{\varrho^{G^\prime}(\nu^\doubleprime_0) - \varrho^{G^\prime}(\nu^{\prime}_0)}_\infty
            =
            \norm[\big]{\pi(\ell^\doubleprime) - \pi(\ell^{\prime})}_\infty
            \le
            \norm[\big]{\ell^\doubleprime - \ell^{\prime}}_\infty
            \le
            \norm[\big]{z^\doubleprime - z^\prime}_\infty
            \le
            C \abs[\big]{\nu^\doubleprime_0 - \nu^{\prime}_0} ,
        \end{equation*}
        i.e., $\varrho^{G^\prime}$ is locally Lipschitz continuous.
        In order to see piecewise linearity, we consider a slight variant of \eqref{eq:linearcomplementarityproblem} where $\nu^\prime_0$ is regarded a variable and $\ell^\prime_0$ is fixed to one, i.e.,
        \begingroup
        \renewcommand{\arraystretch}{1.25}
        \begin{equation}
            z \ge 0,
            \qquad
            \begin{pmatrix}
              0 & 0 \\
              -\one_t & M
            \end{pmatrix} z - \one_s \ge 0,
            \qquad
            z^\top \left( \begin{pmatrix}
              0 & 0 \\
              -\one_t & M
            \end{pmatrix} z - \one_s \right) = 0 .
          \tag{\ensuremath{\text{LCP}^\prime}}
          \label{eq:linearcomplementarityproblem:parametric}
        \end{equation}
        \endgroup
        Its set of solutions corresponds exactly to all sets of solutions to \eqref{eq:linearcomplementarityproblem} with arbitrary flow rate $\nu^\prime_0 \ge 0$ and $\ell^\prime_0 = 1$.
        Then using again \cref{lemma:normalizedthinflow:normalization,corollary:normalizedthinflow:linearcomplementarityproblem}, the hypograph of $\varrho^{G^\prime}$ can be written as
        \begingroup
        \setlength{\jot}{.5em}
        \begin{align*}
            \hyp\big( \varrho^{G^\prime} \big)
            &=
            \graph\big( \varrho^{G^\prime} \big) + \big( \set{0} \times \Rn^V \big)
        \\
            &=
            \set[\Big]{\big( \nu_0^\prime, \varrho^{G^\prime}(\nu_0^\prime) \big) \in \R \times \R^V \colon \nu_0^\prime \ge 0} + \big( \set{0} \times \Rn^V \big)
        \\
            &=
            \set[\Big]{\big( \nu_0^\prime, \pi(\ell^\prime) \big) \in \R \times \R^V \colon (\nu^\prime_0, \ell^\prime, x^\prime, y^\prime) \text{ solves \eqref{eq:linearcomplementarityproblem:parametric}}} + \big( \set{0} \times \Rn^V \big)
        \\
            &=
            \set[\Big]{\big( \nu^\prime_0, \ell^\prime\big) \in \R \times \R^V \colon (\nu^\prime_0, \ell^\prime, x^\prime, y^\prime) \text{ solves \eqref{eq:linearcomplementarityproblem:parametric}}} + \big( \set{0} \times \Rn^V \big)
        \end{align*}
        \endgroup
        By \citet*[p. 646]{CPS2009}, the first (Minkowski) summand of the right hand side is the union of finitely many polyhedra as it is the linear projection of the set of solutions to a linear complementarity problem.
        Hence, the same holds for $\hyp\big( \varrho^{G^\prime} \big)$ and $\graph\big( \varrho^{G^\prime} \big)$.
        This yields that $\varrho^{G^\prime}$ must be piecewise linear.

        We would like to define $\chi^{G^\prime}\colon \Rp \to \Rp^{A^\prime}$ in a similar way as $\varrho^{G^\prime}$ such that, for every $\nu_0^\prime \ge 0$, the vector $\chi^{G^\prime}\!(\nu_0^\prime)$ is a normalized thin $s$--$t$-flow with resetting of value $\nu_0^\prime$ in $G^\prime$ that admits the corresponding label $1$ at $s$.
        As \cref{example:normalizedthinflow:nonuniqueness} shows, the flow values on the arcs are generally not unique.
        $\chi^{G^\prime}$ can be chosen in any way such that its graph lies within the projection of the set of solutions to \eqref{eq:linearcomplementarityproblem:parametric} onto the flow rate $\nu_0^\prime$ and flow variables $x^\prime$.
        Reasoning as above, it follows that this can be done such that $\chi^{G^\prime}$ is piecewise linear and continuous as well.
    \end{proof}
    \pagebreak[3]
    \begin{lemma}
        \label{lem:normalizedthinflow:parametric:monotonicity}
        Let $G^\prime = (V, A^\prime, A^*)$ be a shortest path graph with resetting.
        \begin{enumerate}[label=\roman*),ref=(\roman*)]
            \item \label{lem:normalizedthinflow:parametric:monotonicity:scaling}
            For all $\nu_0^\prime, \ell^\prime_0 \ge 0$, the vector $\ell^\prime_0 \cdot \chi^{G^\prime}\!\big(\tfrac{\nu_0^\prime}{\ell^\prime_0}\big)$ is a normalized thin $s$--$t$-flow with resetting of value~$\nu_0^\prime$ in $G^\prime$ and $\ell^\prime_0 \cdot \varrho^{G^\prime}\!\big(\tfrac{\nu_0^\prime}{\ell^\prime_0}\big)$ are corresponding labels.
            \item \label{lem:normalizedthinflow:parametric:monotonicity:monotonicity}
            For all $v \in V$, the function $\varrho^{G^\prime}_v$ is monotonically non-decreasing and $\nu_0^\prime \mapsto \frac{1}{\nu_0^\prime} \varrho^{G^\prime}_v(\nu_0^\prime)$ is monotonically non-increasing.
        \end{enumerate}
    \end{lemma}

    \begin{proof}
        Consulting \cref{theorem:normalizedthinflow:linearcomplementarityproblem}, it becomes quite obvious that a normalized thin flow with resetting with corresponding label $1$ at $s$ and flow value~$\tfrac{\nu_0^\prime}{\ell^\prime_0}$ can be scaled by $\ell^\prime_0$ to yield a normalized thin flow with resetting with corresponding label $\ell^\prime_0$ at $s$ and flow value $\nu_0^\prime$.
        This yields the first statement.

        The monotonicity of $\varrho^{G^\prime}$ is implied directly by \cref{theorem:normalizedthinflow:monotonicity}.
        By the above, $\tfrac{1}{\nu_0^\prime} \varrho^{G^\prime}(\nu_0^\prime)$ can be interpreted as the corresponding labels of a normalized thin flow of value $1$ with label $\tfrac{1}{\nu_0^\prime}$ at $s$.
        Its monotonicity follows again from \cref{theorem:normalizedthinflow:monotonicity}.
    \end{proof}

    We finish this section by arguing that the normalization constraints of \cref{theorem:normalizedthinflow:linearcomplementarityproblem} \ref{theorem:normalizedthinflow:linearcomplementarityproblem:lcp} can be incorporated in \eqref{eq:linearcomplementarityproblem}.
    The resulting linear complementarity problem is an exact formulation for normalized thin flows with resetting.
    This extension, however, requires a rather technical construction, which is why the presentation with successive normalization was chosen.
    In order to model the normalization constraints for $w \in V$ with $\delta^-(w) \cap A^* = \emptyset$,
    fix an arbitrary ordering $N^-(w) = \set{v_1, \ldots, v_{n}}$ where $n = \card{N^-(w)}$.
    Introduce new variables $d_{w, i}$ for $i \in [n]$ and add the complementarity conditions
    \begingroup
    \setlength{\jot}{.5em}
    \begin{align*}
        d_{w, 1} &\ge 0,
        &
        \ell^\prime_w - \ell^\prime_{v_{n}} + d_{w, n} &\ge 0,
        &
        d_{w, 1} \big(\ell^\prime_w - \ell^\prime_{v_{n}} + d_{n}\big) &=0
    \\
        d_{w, 2} &\ge 0,
        &
        \ell^\prime_{v_{1}} - \ell^\prime_{v_{2}} + d_{w, 2} &\ge 0,
        &
        d_{w, 2} \big(\ell^\prime_{v_{1}} - \ell^\prime_{v_{2}} + d_{w, 2}\big) &= 0
        &
        &
    \\
        d_{w, i} &\ge 0,
        &
        \ell^\prime_{v_{i - 1}} - d_{w, i - 1} - \ell^\prime_{v_{i}} + d_{w, i} &\ge 0,
        &
        d_{w, i} \big(\ell^\prime_{v_{i - 1}} - d_{w, i - 1} - \ell^\prime_{v_{i}} + d_{w, i}\big) &= 0
        &
        &\text{for all } 3 \le i \le n .
    \end{align*}
    \endgroup
    It can be shown by induction that these conditions are equivalent to $\ell^\prime_{v_i} - d_{w, i} = \min_{j \in [i]} \ell^\prime_{v_j}$ for all $i \in [n]$ and $\ell^\prime_w \ge \ell^\prime_{v_n} - d_{w, n}  = \min_{j \in [n]} \ell^\prime_{v_j} = \min_{v \in N^-(w)} \ell^\prime_v$.

    \section{Evolution of dynamic equilibria}
    \label{section:evolution}

    We extend the constructive method for dynamic equilibria by \citet{KS2011} to inflow rates $\nu_0 \in L_{loc}^{1}(\R)$ which are right-monotone (in addition to being non-negative and vanishing almost everywhere on $\Rnn$).
    A definition of right-monotonicity follows the next theorem which holds for arbitrary inflow rates.

    Consider the earliest times function $\ell\colon \R \to \Rp^V$ of a dynamic equilibrium.
    For $\vartheta \in \R$, the sets of active and resetting arcs agree with $A^\prime_\vartheta$ and $A^*_\vartheta$ and, therefore, are determined by $\ell(\vartheta)$ only.
    If $\ell$ is right-differentiable at $\vartheta \in \R$, its right-derivative $\tfrac{\text{d} \ell}{\text{d} \vartheta^+}$ at $\vartheta$ represents the corresponding labels of a normalized thin flow with resetting of value $\nu_0(\vartheta)$ in $G^\prime_{\vartheta}$.
    Hence, also $\tfrac{\text{d} \ell}{\text{d} \vartheta^+}(\vartheta)$ is determined by~$\ell(\vartheta)$.
    Indeed, the earliest times function of dynamic equilibria can be characterized as the set of solutions to a differential equation, as the next theorem states.

    In the case that $\ell$ is the earliest times function of a dynamic equilibrium, we know that $G^\prime_\vartheta$ is a shortest path graph with resetting for all $\vartheta \in \R$.
    To ensure this property through the differential equation, we need to extend the definition of $\varrho^{G^\prime}$ and $\chi^{G^\prime}$ to triples $G^\prime = (V, A^\prime, A^*)$ such that $A^* \subseteq A^\prime \subseteq A$ where $A^\prime$ is acyclic, but there are $v \in V$ without an $s$--$v$-path in $A^\prime$.
    In that case, let $U := \set{v \in V\colon \exists\, \text{$s$--$v$-path in } A^\prime}$, set $\varrho^{G^\prime}_v \equiv 0$ for $v \in U$, and $\varrho^{G^\prime}_v \equiv 1$ for $v \in V \setminus U$.
    Further, define $\chi^{G^\prime}_a \equiv 0$ for all $a \in A^\prime$.

    \begin{theorem}[Differential equation]
        \label{thm:differentialequation}
        Let $\ell^0_v$ denote the shortest distance from $s$ to $v$ in $G$ with respect to~$\tau$ for all $v \in V$.
        Then $\ell\colon \R \to \Rp^V$ is the earliest times function of a dynamic equilibrium if and only if $\ell$ is a locally absolutely continuous solution to
        \begin{equation*}
            \label{eq:differentialequation}
            \tag{DE}
            \ell(\vartheta) = \big(\vartheta + \ell^0_v\big)_{v \in V} \text{ for all } \vartheta \le 0 \quad \text{and} \quad \tfrac{\text{d} \ell}{\text{d} \vartheta} (\vartheta) = \varrho^{G^\prime_\vartheta} \big( \nu_0(\vartheta) \big) \text{ for a.e.\ } \vartheta \ge 0 .
        \end{equation*}
    \end{theorem}

    \begin{proof}
        Assume $\ell$ is the earliest times function of a dynamic equilibrium.
        Since the network is assumed to be empty up to time zero, the initial condition on $\Rn$ holds due to \eqref{eq:bellmanequation}.
        The components of $\ell$ are locally absolutely continuous and, hence, differentiable almost everywhere.
        \citet{CCL2015} show that, if it exists, the right-derivative of $\ell$ at time $\vartheta \in \R$ represents the corresponding labels of a normalized thin flow with resetting of value $\nu_0(\vartheta)$ in $G^\prime_\vartheta$ with label $1$ at~$s$.
        Hence, $\tfrac{\text{d} \ell}{\text{d} \vartheta} (\vartheta) = \varrho^{G^\prime_\vartheta} \big( \nu_0(\vartheta) \big)$ follows for almost every $\vartheta \ge 0$.

        \bigskip
        For the converse direction, assume $\ell$ is a locally absolutely continuous solution to the differential equation \eqref{eq:differentialequation}.
        $\ell$ defines $G^\prime_{\vartheta} = (V, A^\prime_\vartheta, A^*_\vartheta)$ for all $\vartheta \in \R$.
        To simplify notation, we set
        \begin{equation*}
            \varrho\colon \R \to \Rp^V, \vartheta \mapsto \varrho^{G^\prime_\vartheta}\big(\nu_0(\vartheta)\big)
            \quad \text{ and } \quad
            \chi\colon \R \to \Rp^A, \vartheta \mapsto \begin{cases}
                 \chi^{\smash{G^\prime_\vartheta}}_a\big(\nu_0(\vartheta)\big) &\text{for } a \in A^\prime_\vartheta \\
                0 &\text{for } a \not\in A^\prime_\vartheta .
            \end{cases}
        \end{equation*}
        Note that in both definitions the shortest path network $G^\prime_\vartheta$ as well as the inflow rate $\nu_0(\vartheta)$ depend on the parameter $\vartheta$.
        The codomain of $\chi^{\smash{G^\prime_\vartheta}}$ depends on $\vartheta$.
        $\chi$ extends the definition of the flows on $A^\prime_\vartheta$ onto $A$ by zero.
        The fundamental theorem of calculus yields
        \begin{equation*}
            \ell(\vartheta) = \ell^0 + \int_{0}^{\vartheta} \varrho(\theta) \d\theta \qquad \text{for all } \vartheta \in \R,
        \end{equation*}
        where integration is applied element-wise.
        Since $\varrho_v \ge 0$ for all $v \in V$, $\ell_v$ is a monotonically non-decreasing function.
        Based on $\ell$, we define a flow over time and prove step by step that it is feasible, has the earliest times function $\ell$, and satisfies the conditions of a dynamic equilibrium.

        \medskip
        \emph{The shortest path graphs with resetting.}
        We start by showing that $G^\prime_\vartheta$ is a shortest path graph with resetting for all $\vartheta \in \R$.
        It is clear from the initial condition that the claim holds for $\vartheta \le 0$.
        Assume for contradiction, that there is $\vartheta > 0$ such that $A^\prime_\vartheta$ is not a shortest path graph with resetting.
        This means that the set $U := \set{v \in V\colon \exists\, \text{$s$--$v$-path in } A^\prime_\vartheta}$ is a proper subset of $V$ and $\delta^+(U) \cap A^\prime_\vartheta = \emptyset$.
        Let $\ubar{\vartheta} := \sup \set{\theta \le \vartheta\colon \delta^+(U) \cap A^\prime_\theta \ne \emptyset}$ be the last time before $\vartheta$ that an arc leaving $U$ was active.
        Then, $0 \le \ubar{\vartheta} < \vartheta$ due to the initial condition and the continuity of $\ell$.
        In particular, it holds $\delta^+(U) \cap A^\prime_{\ubar{\vartheta}} \ne \emptyset$, and $\delta^+(U) \cap A^\prime_\theta = \emptyset$ for all $\theta \in (\ubar{\vartheta}, \vartheta]$.
        Then for all $(v, w) \in \delta^+(U) \cap A^\prime_{\ubar{\vartheta}}$, $\varrho_w \equiv 1$ and $\varrho_w - \varrho_v \ge 0$ on $(\ubar{\vartheta}, \vartheta]$  yield $\emptyset \ne \delta^+(U) \cap A^\prime_{\ubar{\vartheta}}  \subseteq \delta^+(U) \cap A^\prime_\vartheta = \emptyset$, a contradiction.

        As an immediate consequence, $\lim_{\vartheta \to +\infty} \ell_v(\vartheta) \ge \lim_{\vartheta \to +\infty} \ell_s(\vartheta) + \ell^0_v = +\infty$ holds for all $v \in V$.
        On the other hand, $\lim_{\vartheta \to -\infty} \ell_v(\vartheta) = \lim_{\vartheta \to -\infty} \vartheta + \ell^0_v = -\infty$.
        The continuity of $\ell$ yields $\ell_v(\R) = \R$ for all $v \in V$.

        \medskip
        \emph{The flow over time.}
        For $a = (v, w) \in A$, define the flow functions $f_a^+\colon \R \to \Rp$ and $f_a^-\colon \R \to \Rp$ by setting
        \begin{equation*}
            (f^+_a \circ \ell_v) \cdot \varrho_v \equiv \chi_a \equiv (f^-_a \circ \ell_w) \cdot \varrho_w
            \quad \text{ a.e.\ on } \R .
        \end{equation*}
        We shall see that this is a sound definition.
        For $\ubar{\vartheta} < \bar{\vartheta}$ such that $\ell_v(\ubar{\vartheta}) = \ell_v(\bar{\vartheta})$, the monotonicity of $\ell_v$ implies $\ell_v \equiv \ell_v(\ubar{\vartheta})$ and $\varrho_v \equiv 0$ almost everywhere on $[\ubar{\vartheta}, \bar{\vartheta}]$.
        Hence, $f^+_a$ is well-defined at almost every $\ell_v(\vartheta)$ for $\vartheta \in \R$ such that $\tfrac{\text{d} \ell_v}{\text{d}\vartheta}(\vartheta) = \varrho_v(\vartheta) \ne 0$.
        As the set $\set{\ell_v(\vartheta)\colon \vartheta \in \R \text{ and } \tfrac{\text{d} \ell_v}{\text{d}\vartheta}(\vartheta) = 0}$ has measure zero, $f^+_a$ is well-defined almost everywhere on $\ell_v(\R) = \R$.
        Reasoning in a similar way gives that $f^-_a$ is also well-defined almost everywhere on $\R$.
        \\
        \Cref{definition:normalizedthinflow} yields almost everywhere $f^-_a \le \nu_a$.
        Similarly, almost everywhere $f^+_a \equiv f^+_a \circ \ell_s \equiv \tfrac{\chi_a}{\varrho_s} \le \nu_0$ if $v = s$ and $f^+_a \le \sum_{b \in \delta^-(v)} \nu_{b}$ if $v \ne s$.
        As $\nu_0$ is locally integrable and $(\nu_{b})_{{b} \in A}$ is constant, the functions $f^+_a$ and $f^-_a$ are locally integrable as well.
        Further, $f^+_a$ and $f^-_a$ vanish on $\R_{\le 0}$ as $\nu_0$ does so and $\ell(0) = \ell^0 \ge 0$.
        The flow conservation constraints hold, since normalized thin flows with resetting obey them.

        \medskip
        \emph{The queue lengths.}
        For $a = (v, w) \in A$, we define the function $z_a$ via its derivative and show that it evaluates to the queue length of $a$ as induced by $f^+_a$.
        For that purpose, we define $\tfrac{\text{d} z_a}{\text{d} \vartheta}$ by setting
        \begin{equation*}
            \tfrac{\text{d} z_a}{\text{d} \vartheta}\big( \ell_v(\vartheta) \big) \cdot \varrho_v(\vartheta)
            = \begin{cases}
                0 &\text{if } a \in A \setminus A^\prime_\vartheta \\
                \nu_a \cdot \pos[\big]{ \varrho_w(\vartheta) - \varrho_v(\vartheta) } &\text{if } a \in A^\prime_\vartheta \setminus A^*_\vartheta \\
                \nu_a \cdot \big( \varrho_w(\vartheta) - \varrho_v(\vartheta) \big) &\text{if } a \in A^*_\vartheta
            \end{cases}
            \qquad
            \text{for a.e.\ } \vartheta \in \R .
        \end{equation*}
        As argued before, this determines $\tfrac{\text{d} z_a}{\text{d} \vartheta}$ almost everywhere on $\R$ as $\set{\ell_v(\vartheta)\colon \vartheta \in \R \text{ and } \tfrac{\text{d} \ell_v}{\text{d}\vartheta}(\vartheta) = 0}$ has measure zero.
        Using $0 \le \varrho_w \le \max \set[\big]{\varrho_v, \tfrac{\chi_a}{\nu_a}}$, we can bound $\tfrac{\text{d} z_a}{\text{d} \vartheta}$ almost everywhere on $\ell_v(\R) = \R$ by noting that for almost every $\vartheta \in \R$
        \begin{align*}
          \abs{\tfrac{\text{d} z_a}{\text{d} \vartheta}\big( \ell_v(\vartheta) \big)} \cdot \varrho_v(\vartheta)
          \le
          \nu_a \cdot \abs{ \varrho_w(\vartheta) - \varrho_v(\vartheta) }
          \le
          \max \set[\big]{\nu_a \cdot \varrho_v(\vartheta), \chi_a(\vartheta)}
          =
          \max \set{ \nu_a, f^+_a\big(\ell_v(\vartheta)\big) } \cdot \varrho_v(\vartheta) .
        \end{align*}
        This shows that the local integrability of $f^+_a$ carries over to $\tfrac{\text{d} z_a}{\text{d} \vartheta}$.
        Therefore, $\tfrac{\text{d} z_a}{\text{d} \vartheta}$ is indeed the derivative of the function $z_a\colon \R \to \Rp, \vartheta \mapsto \int_{\ell_v(0)}^\vartheta \tfrac{\text{d} z_a}{\text{d} \vartheta} (\theta) \d\theta$.
        Since $\ell_v, \ell_w$, and $z_a$ are locally absolutely continuous and $\ell_v$ is monotone, the formula for the change of variables applies to $z_a \circ \ell_v$ and yields
        \begin{equation*}
          z_a \big( \ell_v(\vartheta) \big) = \int_{0}^\vartheta \tfrac{\text{d} z_a}{\text{d} \vartheta}\big( \ell_v(\theta) \big) \cdot \varrho_v(\theta) \d\theta = \int_0^\vartheta \nu_a \cdot \tfrac{\text{d}}{\text{d} \theta} \big[ \ell_w(\theta) - \ell_v(\theta) - \tau_a \big]_+ \d\theta = \nu_a \cdot \big[ \ell_w(\vartheta) - \ell_v(\vartheta) - \tau_a \big]_+ .
        \end{equation*}

        Now, we can show that $z_a$ is indeed the queue length that is induced by $f^+_a$ by proving that $z_a$ is a solution to \eqref{eq:queuedynamics}.
        The above identity implies $z_a \equiv 0$ on $\R_{\le 0} \subseteq (-\infty, \ell^0_v]$.
        Let $\vartheta \in \R$ such that $\varrho_v(\vartheta) \ne 0$.
        In the case $a \in A \setminus A^\prime_\vartheta$, we know from the above identity and continuity of $\ell$ that $z_a\big(\ell_v(\vartheta)\big) = 0$ and $\tfrac{\text{d} z_a}{\text{d} \vartheta} \big( \ell_v(\vartheta) \big) = 0$.
        On the other hand, the definition of $\chi$ yields $\chi_a(\vartheta) = 0$, which implies $\pos{f^+_a\big(\ell_v(\vartheta)\big) - \nu_a} = 0$ and, thus, $\eqref{eq:queuedynamics}$ holds for $a$ at time $\ell_v(\vartheta)$.
        If $a \in A^\prime_\vartheta \setminus A^*_\vartheta$, we also know $z_a\big( \ell_v(\vartheta) \big) = 0$.
        In the case that $\nu_a \cdot \varrho_w(\vartheta) = \max \set[\big]{\nu_a \cdot \varrho_v(\vartheta), \chi_a(\vartheta)}$, it follows immediately that $\tfrac{\text{d} z_a}{\text{d} \vartheta}\big(\ell_v(\vartheta)\big) = \pos{f^+_a\big(\ell_v(\vartheta)\big) - \nu_a}$.
        Otherwise, $\chi_a(\vartheta) = 0$ must holds, which implies the equation $\tfrac{\text{d} z_a}{\text{d} \vartheta}\big(\ell_v(\vartheta)\big) = 0 = \pos{f^+_a\big(\ell_v(\vartheta)\big) - \nu_a}$.
        Finally, if $a \in A^*_\vartheta$, then $z_a\big(\ell_v(\vartheta)\big) > 0$ holds.
        Further, $\nu_a \cdot \varrho_w(\vartheta) = \chi_a(\vartheta)$ yields $\tfrac{\text{d} z_a}{\text{d} \vartheta}\big(\ell_v(\vartheta)\big) = f^+_a\big(\ell_v(\vartheta)\big) - \nu_a$.
        In total, $z_a$ fulfills \eqref{eq:queuedynamics} for almost every $\ell_v(\vartheta)$ such that $\vartheta \in \R$ and $\varrho_v(\vartheta) \ne 0$, which is almost everywhere on $\R$.

        \medskip
        \emph{The queuing dynamics.}
        To see that $(f^+, f^-)$ is a feasible flow over time, we will show that it respects the queuing dynamics $z_a\big(\ell_v(\vartheta)\big) = F^+_a\big(\ell_v(\vartheta)\big) - F^-_a\big(\ell_v(\vartheta) + \tau_a\big)$ for all $a \in A$ and $\vartheta \in \R$.

        Let $a = (v, w) \in A$.
        We start by establishing it for a single point in time and extend it by looking at its derivative.
        Consider the time $\hat{\vartheta} = \inf \set{\theta \ge 0\colon a \in A^\prime_\theta}$, relative to which $a$ is active for the first time.
        If no such time exists, i.e., $\hat{\vartheta} = +\infty$, then $z_a \equiv 0$, $F_a^+ \equiv 0$, and $F_a^- \equiv 0$ satisfy the equation trivially.
        Otherwise, we get $\ell_w(\hat{\vartheta}) = \ell_v(\hat{\vartheta}) + \tau_a$ and, therefore,
        \begin{equation*}
            z_a\big(\ell_v(\hat{\vartheta})\big) = 0 = F^+_a\big(\ell_v(\hat{\vartheta})\big) - F^-_a\big(\ell_w(\hat{\vartheta})\big) = F^+_a\big(\ell_v(\hat{\vartheta})\big) - F^-_a\big(\ell_v(\hat{\vartheta}) + \tau_a\big) .
        \end{equation*}
        It remains to prove that the derivatives of both sides of the equation agree almost everywhere.
        As we have shown that $z_a$ is a solution to \eqref{eq:queuedynamics}, it suffices to prove for almost every $\vartheta \in \R$ with $\varrho_v(\vartheta) \ne 0$ that
        \begin{equation*}
            f^-_a\big(\ell_v(\vartheta) + \tau_a\big)
            =
            \begin{cases}
                \min \set[\big]{ f^+_a\big(\ell_v(\vartheta)\big), \nu_a }
                &\text{if } z_a\big(\ell_v(\vartheta)\big) = 0, \text{ i.e., } a \in A \setminus A^*_\vartheta
            \\
                \nu_a
                &\text{if } z_a\big(\ell_v(\vartheta)\big) > 0, \text{ i.e., }a \in A^*_\vartheta .
            \end{cases}
        \end{equation*}

        First, consider $\vartheta \in \R$ such that $a \in A^\prime_\vartheta \setminus A^*_\vartheta$, i.e., $\ell_w(\vartheta) = \ell_v(\vartheta) + \tau_a$.
        Due to the continuity of $\ell$, the set $\set{\theta \in \R\colon a \in A^\prime_\theta \setminus A^*_\theta}$ consists of closed intervals.
        Consequently, for almost every considered $\vartheta$ it holds $\varrho_w(\vartheta) = \varrho_v(\vartheta)$.
        Therefore, $f^-_a\big(\ell_v(\vartheta) + \tau_a\big) = f^-_a\big(\ell_w(\vartheta)\big) = f^+_a\big(\ell_v(\vartheta)\big)$ and $f^-_a\big(\ell_w(\vartheta)\big) \le \nu_a$.
        Hence, the equation holds for almost all $\vartheta \in \R$ such that $\varrho_v(\vartheta) \ne 0$ and $a \in A^\prime_\vartheta \setminus A^*_\vartheta$.

        If $a \in A \setminus A^\prime_\vartheta$, define $\ubar{\vartheta} := \sup \set{\theta \le \vartheta\colon a \in A^\prime_\theta} \in \R \cup \set{-\infty}$ and $\bar{\vartheta} := \inf \set{\theta \ge \vartheta\colon a \in A^\prime_\theta} \in \R \cup \set{+\infty}$.
        Continuity of $\ell$ yields $a \in A \setminus A^\prime_\theta$ for all $\theta \in (\ubar{\vartheta}, \bar{\vartheta})$ as well as $\ell_w(\ubar{\vartheta}) = \ell_v(\ubar{\vartheta}) + \tau_a$ and $\ell_w(\bar{\vartheta}) = \ell_v(\bar{\vartheta}) + \tau_a$.
        Therefore, monotonicity and continuity of $\ell_v$ and $\ell_w$ imply that $f^+_a$ and $f^-_a$ vanish almost everywhere on the intervals $\big(\ell_v(\ubar{\vartheta}), \ell_v(\bar{\vartheta})\big)$ and $\big(\ell_w(\ubar{\vartheta}), \ell_w(\bar{\vartheta})\big) = \big(\ell_v(\ubar{\vartheta}) + \tau_a, \ell_v(\bar{\vartheta}) + \tau_a\big)$, respectively.
        This proves that the equation is fulfilled for almost every $\vartheta \in \R$ such that $a \in A \setminus A^\prime_\vartheta$.

        A similar argument shows that $f^-_a\big(\ell_v(\vartheta) + \tau_a\big) = \nu_a$ for almost every $\vartheta$ with $a \in A^*_\vartheta$.
        In conclusion, $z_a(\vartheta) = F^+_a(\vartheta) - F^-_a(\vartheta + \tau_a)$ holds on $\R$.

        \medskip
        \emph{The earliest times function.}
        Now, it is immediate to see that $\ell$ indeed is the earliest times function of the feasible flow over time $(f^+, f^-)$.
        For $a = (v, w) \in A$ and $\vartheta \in \R$, the above characterization of $z_a$ in terms of $\ell$ immediately implies $\ell_w(\vartheta) \le \ell_v(\vartheta) + \tfrac{1}{\nu_a} \cdot z_a\big(\ell_v(\vartheta)\big) + \tau_a$ with equality if and only if $a \in A^\prime_\vartheta$.
        This proves that $\ell$ satisfies the Bellman equations \eqref{eq:bellmanequation}.

        \medskip
        \emph{The equilibrium condition.}
        To complete the proof, we only have to show that $(f^+, f^-)$ is a dynamic equilibrium.
        For $a = (v, w) \in A$, $\ell_v$ and $\ell_w$ are monotonically non-decreasing.
        Hence, for every $\vartheta \in \R$, two changes of variables yield
        \begin{equation*}
            F^+_a\big( \ell_v(\vartheta) \big)
            =
            \int_{0}^{\vartheta} f^+_a\big( \ell_v(\theta) \big) \cdot \varrho_v \big( \theta \big) \d\theta
            =
            \int_{0}^{\vartheta} \chi_a \big( \theta \big) \d\theta
            =
            \int_{0}^{\vartheta} f^-_a\big( \ell_w(\theta) \big) \cdot \varrho_w \big( \theta \big) \d\theta
            =
            F^-_a\big( \ell_w(\vartheta) \big)
            .
        \end{equation*}
        \Cref{lem:dynamicequilibrium:characterization} gives that $(f^+, f^-)$ is a dynamic equilibrium.
    \end{proof}

    The above theorem suggests to construct a dynamic equilibrium by integrating over thin flows with resetting.
    The method of \citet{KS2011} does this by implicitly assuming that $\ell$ is right-linear.
    This is feasible for piecewise constant inflow rates, as there always is a dynamic equilibrium with that property.
    For these dynamic equilibria, the functions $\vartheta \mapsto \varrho^{\smash{G^\prime_\vartheta}}_v \big(\nu_0(\vartheta)\big)$ are right-constant.
    We want to consider a more general class of inflow rates.
    \begin{definition}[Monotone functions]
        We call a function $g \in L^1_\text{loc}(\R)$ \emph{monotonically non-decreasing} (\emph{non-increasing}) if there exists a set $N \subseteq \R$ of measure zero such that $g(\xi) \le g(\hat{\xi})$ $\big(g(\xi) \ge g(\hat{\xi})\big)$ for all $\xi, \hat{\xi} \in \R \setminus N$ with $\xi \le \hat{\xi}$.
        \\
        We call a function $g$ \emph{monotone} if it is monotonically non-decreasing or monotonically non-increasing.
        Further, $g$ is \emph{right-monotone} (\emph{left-monotone}) if for every $\xi \in \mathbb{R}$ there is $\varepsilon > 0$ such that $g$ is monotone on $[\xi, \xi + \varepsilon]$ $\big([\xi - \varepsilon, \xi]\big)$.
    \end{definition}

    For a right-monotone inflow rate $\nu_0$, the map $\vartheta \mapsto \varrho^{\smash{G^\prime_\vartheta}}_v \big(\nu_0(\vartheta)\big)$ cannot be expected to be right-constant.
    Due to the piecewise linear dependency of the thin flows with resetting on the flow value, however, this map is right-monotone.
    This still allows to use the same method for constructing a dynamic equilibrium as follows.

    \begin{theorem}[$\alpha$-extension of dynamic equilibria]
        \label{dynamicequilibrium:extension}
        Let the inflow rate $\nu_0$ be right-monotone.
        For $\vartheta \ge 0$, let $\ell\colon (-\infty, \vartheta] \to \Rp^V$ fulfill the differential equation \eqref{eq:differentialequation} on $(-\infty, \vartheta]$.
        Then there is $\alpha > 0$ such that $\ell$ can be extended to fulfill it on $(-\infty, \vartheta + \alpha]$.
    \end{theorem}

    For the proof of this theorem, we make use of some basic properties of right-monotone functions which we show first.
    \cref{lem:monotone:composition} regards the composition of left-/right-monotone functions.
    \cref{lem:monotone:compatibility} relates right-monotone functions to their primitives.

    \begin{lemma}[Composition of left-/right-monotone functions]
        \label{lem:monotone:composition}
        Let $h \in L^1_\text{loc}(\R)$ be right-monotone and $g \in L^1_\text{loc}(\R)$ be left- and right-monotone such that their composition $g \circ h$ is well-defined.
        If $h$ is locally bounded or there exists $\upsilon \ge 0$ such that $g$ is monotone on the unbounded intervals $(-\infty, -\upsilon)$ and $(\upsilon, +\infty)$, then also $g \circ h$ is right-monotone.
    \end{lemma}

    \begin{proof}
        Let $\xi \in \R$.
        Assume $h$ is monotonically non-decreasing on $(\xi, \xi + \varepsilon)$ for some $\varepsilon > 0$.
        Set $\ubar{\upsilon}$ to be the essential infimum of $h$ on $(\xi, \xi + \varepsilon)$, i.e.,
        \begin{equation*}
            \ubar{\upsilon} := \essinf h\vert_{(\xi, \xi + \varepsilon)} = \sup \set{\upsilon \in \R\colon h \ge \upsilon \text{ a.e.\ on } (\xi, \xi + \varepsilon) } .
        \end{equation*}
        If $\ubar{\upsilon} > -\infty$, there is $\bar{\upsilon} > \ubar{\upsilon}$ such that $g$ is monotone on $(\ubar{\upsilon}, \bar{\upsilon})$ as $g$ is right-monotone.
        Otherwise, $h$ is not locally bounded and the assumption gives $\bar{\upsilon} \in \R$ such that $g$ is montone on $(\ubar{\upsilon}, \bar{\upsilon}) = (-\infty, \bar{\upsilon})$.
        Choose $0 < \hat{\varepsilon} \le \varepsilon$ small enough such that $h \le \bar{\upsilon}$ almost everywhere on $(\xi, \xi + \hat{\varepsilon})$.
        Then $g \circ h$ is monotone on $(\xi, \xi + \hat{\varepsilon})$.

        The case that $h$ is monotonically non-increasing on $(\xi, \xi + \varepsilon)$ works similarly.
        It needs the left-monotonicity instead of right-monotonicity of $g$.
    \end{proof}

    \begin{lemma}
        \label{lem:monotone:compatibility}
        Let $g \in L^1_\text{loc}(\R)$ be a right-monotone function such that $G\colon \R \to \R, \xi \mapsto \int_0^\xi g(\hat{\xi}) \d\hat{\xi}$ satisfies $\inf \set{\xi > 0\colon G(\xi) > 0} = 0$.
        Then, there is $\varepsilon > 0$ such that $g > 0$ almost everywhere on $(0, \varepsilon)$ and $G > 0$ on $(0, \varepsilon)$.
    \end{lemma}

    \begin{proof}
        Let $\varepsilon > 0$ such that $g$ is monotone on $(0, \varepsilon)$, and set $\delta := \sup\set[\big]{\hat{\delta} \in [0, \varepsilon)\colon g \le 0 \text{ a.e.\ on } [0, \hat{\delta}]}$.
        $\delta > 0$ would imply $G \le 0$ on $[0, \delta]$ which contradicts $\inf \set{\xi > 0\colon G(\xi) > 0} = 0$.
        Hence, $\delta = 0$.
        If $g$ is monotonically non-decreasing on $(0, \varepsilon)$, then $g > 0$ almost everywhere on $(0, \varepsilon)$ and $G > 0$ on $(0, \varepsilon)$ follow.
        If $g$ is monotonically non-increasing on $(0, \varepsilon)$, then $\delta = 0$ implies $\esssup g\vert_{(0, \varepsilon)} > 0$.
        Thus, there is $0 < \hat{\varepsilon} \le \varepsilon$ such that $g > 0$ almost everywhere on $(0, \hat{\varepsilon})$ and $G > 0$ on $(0, \hat{\varepsilon})$.
    \end{proof}

    \begin{proof}[Proof of \cref{dynamicequilibrium:extension}.]
        We extend $\ell$ by assuming that $G^\prime_{\vartheta}$ is constant to the right, i.e., $G^\prime_{\vartheta + \varepsilon} = G^\prime_\vartheta$ for small $\varepsilon > 0$.
        This assumption is not necessarily true.
        We will see, however, that the extension we get in this way fulfills \eqref{eq:differentialequation}.
        For that purpose, define for $\varepsilon > 0$
        \begin{equation*}
            \ell(\vartheta + \varepsilon) := \ell(\vartheta) + \int_{\vartheta}^{\vartheta + \varepsilon} \varrho^{G^\prime_\vartheta} \big( \nu_0(\theta) \big) \d\theta ,
        \end{equation*}
        where integration is applied element-wise.
        First of all, we show that this definition is sound.
        For component $v \in V$, the integrand $\varrho^{\smash{G^\prime_\vartheta}}_v \circ \nu_0$ is indeed locally integrable since for every compact set $K \subseteq \R$, \cref{lem:normalizedthinflow:parametric} yields
        \begin{align*}
            \int_K \card{\varrho^{\smash{G^\prime_\vartheta}}_v \big( \nu_0(\theta) \big)} \d\theta
            &= \int_{K \cap \inv{\nu_0}([0, 1])} \varrho^{\smash{G^\prime_\vartheta}}_v \big(\nu_0(\theta)\big) \d\theta
            + \int_{K \setminus \inv{\nu_0}([0, 1])} \frac{\varrho^{\smash{G^\prime_\vartheta}}_v\big(\nu_0(\theta)\big)}{\nu_0(\theta)} \cdot \nu_0(\theta) \d\theta
        \\
            &\le \varrho^{\smash{G^\prime_\vartheta}}_v(1) \cdot \int_{K} \max \set[\big]{ 1, \nu_0(\theta) } \d\theta
            < +\infty .
        \end{align*}
        This extension defines $G^\prime_{\vartheta + \varepsilon} = (V, A^\prime_{\vartheta + \varepsilon}, A^*_{\vartheta + \varepsilon})$ for $\varepsilon > 0$.
        To show that $\ell$ fulfills \eqref{eq:differentialequation} on a strictly larger interval than $(-\infty, \vartheta]$, it is sufficient to prove $\varrho^{G^\prime_{\vartheta + \varepsilon}} \circ \nu_0 \equiv \varrho^{G^\prime_{\vartheta}} \circ \nu_0$ for all small enough $\varepsilon > 0$.
        For that purpose, define the following limits of the sets of active and resetting arcs,
        \begin{align*}
            A^\prime &:= \liminf_{\varepsilon \to 0+} A^\prime_{\vartheta + \varepsilon} = \bigcup_{\delta > 0} \bigcap_{\varepsilon \in (0, \delta)} A^\prime_{\vartheta + \varepsilon} = \set[\big]{a = (v, w) \in A\colon \exists\, \delta > 0\colon \ell_w \ge \ell_v + \tau_a \text{ on } (\vartheta, \vartheta + \delta) } \text{ and }
        \\
            A^* &:= \liminf_{\varepsilon \to 0+} A^*_{\vartheta + \varepsilon} = \bigcup_{\delta > 0} \bigcap_{\varepsilon \in (0, \delta)} A^*_{\vartheta + \varepsilon} = \set[\big]{a = (v, w) \in A\colon \exists\, \delta > 0\colon \ell_w > \ell_v + \tau_a \text{ on } (\vartheta, \vartheta + \delta)} .
        \end{align*}
        We show that $G^\prime_{\vartheta + \varepsilon} = (V, A^\prime, A^*)$ for small $\varepsilon > 0$.
        Due to the continuity of $\ell$, we know that the inclusions $A^\prime \subseteq A^\prime_{\vartheta + \varepsilon} \subseteq A^\prime_\vartheta$ and $A^*_\vartheta \subseteq A^* = A^*_{\vartheta + \varepsilon}$ hold for $\varepsilon > 0$ small enough.

        Since $\nu_0$ is right-monotone and $\varrho^{G^\prime_\vartheta}$ is piecewise linear (with finitely many breakpoints), \cref{lem:monotone:composition} implies that $\varrho^{\smash{G^\prime_\vartheta}}_w \circ \nu_0 - \varrho^{\smash{G^\prime_\vartheta}}_v \circ \nu_0$ is right-monotone for all pairs $v, w \in V$.
        \\
        For arcs $a = (v, w) \in A^\prime_\vartheta \setminus A^\prime$, there is a sequence $(\varepsilon_k)_{k \in \N} \subseteq \Rpp$ such that $\lim_{k \to +\infty} \varepsilon_k = 0$ and $a \in A^\prime_\vartheta \setminus A^\prime_{\vartheta + \varepsilon_k}$ for all $k \in \N$.
        As $\ell$ is continuous and $a \in A^\prime_\vartheta$, a limit argument yields $\ell_w(\vartheta) = \ell_v(\vartheta) + \tau_a$.
        Further, for all $k \in \N$, we get
        \begin{equation*}
            \int_\vartheta^{\vartheta + \varepsilon_k} \varrho^{\smash{G^\prime_\vartheta}}_w \big(\nu_0(\theta)\big) - \varrho^{\smash{G^\prime_\vartheta}}_v \big(\nu_0(\theta)\big) \d\theta
            =
            \big(\ell_w(\vartheta + \varepsilon_k) - \ell_v(\vartheta + \varepsilon_k) - \tau_a\big) - \big(\ell_w(\vartheta) - \ell_v(\vartheta) - \tau_a\big)
            < 0 .
        \end{equation*}
        \Cref{lem:monotone:compatibility} implies $\varrho^{\smash{G^\prime_\vartheta}}_w \circ \nu_0 < \varrho^{\smash{G^\prime_\vartheta}}_v \circ \nu_0$ and, hence, $\chi^{\smash{G^\prime_\vartheta}}_a \circ \nu_0 = 0$ almost everywhere on $(\vartheta, \vartheta + \beta_a)$ for some $\beta_a > 0$.
        Further, it implies $a \not\in A^\prime_{\vartheta + \varepsilon}$ for small $\varepsilon > 0$.
        $A^\prime = A^\prime_{\vartheta + \varepsilon}$ follows for small $\varepsilon > 0$.
        \\
        For $a = (v, w) \in A^* \setminus A^*_\vartheta$, there exists a sequence $(\varepsilon_k)_{k \in \N} \subseteq \Rpp$ such that $\lim_{k \to +\infty} \varepsilon_k = 0$ and $a \in A^*_{\vartheta + \varepsilon_k} \setminus A^*_\vartheta$ for all $k \in \N$.
        Continuity of $\ell$ yields $\ell_w(\vartheta) = \ell_v(\vartheta) + \tau_a$.
        Hence, for $k \in \N$
        \begin{equation*}
            \int_\vartheta^{\vartheta + \varepsilon_k} \varrho^{\smash{G^\prime_\vartheta}}_w \big(\nu_0(\theta)\big) - \varrho^{\smash{G^\prime_\vartheta}}_v \big(\nu_0(\theta)\big) \d\theta
            =
            \big(\ell_w(\vartheta + \varepsilon_k) - \ell_v(\vartheta + \varepsilon_k) - \tau_a\big) - \big(\ell_w(\vartheta) - \ell_v(\vartheta) - \tau_a\big)
            > 0 .
        \end{equation*}
        \Cref{lem:monotone:compatibility} implies $\varrho^{\smash{G^\prime_\vartheta}}_w \circ \nu_0 > \varrho^{\smash{G^\prime_\vartheta}}_v \circ \nu_0$ and, hence, $\chi^{\smash{G^\prime_\vartheta}}_a \circ \nu_0 = \nu_a \cdot \varrho^{\smash{G^\prime_\vartheta}}_w \circ \nu_0$ almost everywhere on $(\vartheta, \vartheta + \beta_a)$ for some $\beta_a > 0$.
        \\
        In total, there is an $\alpha > 0$ such that $G^\prime_{\vartheta + \varepsilon} = (V, A^\prime, A^*)$ for all $0 < \varepsilon < \alpha$.
        More importantly, $\varrho^{G^\prime_\vartheta} \big(\nu_0(\vartheta + \varepsilon)\big)$ are the corresponding labels of a normalized thin flow with resetting of value $\nu_0(\vartheta + \varepsilon)$ in $G^\prime_{\vartheta + \varepsilon}$.
        \Cref{corollary:normalizedthinflow:uniqueness} yields $\varrho^{G^\prime_\vartheta} \circ \nu_0 \equiv \varrho^{G^\prime_{\vartheta + \varepsilon}} \circ \nu_0$ on $(\vartheta, \vartheta + \alpha)$.
    \end{proof}

    Just like in the case of constant inflow rate, the $\alpha$-extension can be applied iteratively to construct a dynamic equilibrium.
    Each such extension that is maximal with respect to $\alpha$ is called a \emph{phase} in the evolution of the dynamic equilibrium.
    Note that \cref{dynamicequilibrium:extension} does not state anything about the length of a phase.
    It is still an open question whether for constant inflow rate the extensions can converge to a finite domain, see \cite{CCO2017}.
    If that is the case, the dynamic equilibrium cannot be computed this way in a finite number of steps.
    In theory, we can take the limit of such a converging sequence of $\alpha$-extensions and repeat.
    As done by \citet{CCO2017,GH2019}, Zorn's lemma can be applied to get a dynamic equilibrium on $\R$.

    \begin{theorem}[Existence of dynamic equilibria]
        For every non-negative, right-monotone, locally integrable inflow rate, there exists a dynamic equilibrium.
    \end{theorem}

    \begin{proof}
        Let $\mathcal{L}$ be the set of functions $\ell\colon (-\infty, \vartheta] \to \Rp^V$ which fulfill \eqref{eq:differentialequation} on $(-\infty, \vartheta)$ for some $\vartheta \in \Rp \cup \set{+\infty}$.
        Define the partial order $\preceq$ on $\mathcal{L}$ by setting $\ell \preceq \hat{\ell}$ if $\dom(\ell) \subseteq \dom(\hat{\ell})$ and $\ell = \hat{\ell} \vert_{\dom(\ell)}$.

        Let $(\ell^{(k)})_{k \in K}$ be a chain in $(\mathcal{L}, \preceq)$ indexed by some set $K$ with domains $\dom(\ell^{(k)}) = (-\infty, \vartheta_k]$.
        Set $\vartheta := \sup_{k \in K} \vartheta_k$ and define the function $\ell\colon (-\infty, \vartheta) \to \Rp^V, \theta \mapsto \sup_{k \in K\colon \theta \le \vartheta_k} \ell^{(k)}(\theta)$.
        Note that $\ell\vert_{(-\infty, \vartheta_k)} \equiv \ell^{(k)}$ for all $k \in K$.
        We would like to continuously extend $\ell$ to $\hat{\ell}$ onto $(-\infty, \vartheta]$.
        Therefore, set $\hat{\ell}(\vartheta) := \lim_{k \to +\infty} \ell(\vartheta_k)$.
        This limit exists as $\ell$ is monotonically non-decreasing.
        It is not clear, however, that it is finite.
        Since $\ell$ fulfills \eqref{eq:differentialequation} on $(-\infty, \vartheta)$, applying \cref{lem:normalizedthinflow:parametric} yields
        \begin{align*}
            \hat{\ell}_v(\vartheta)
            &= \lim_{k \to +\infty} \int_0^{\vartheta_k} \varrho^{\smash{G^\prime_\theta}}_v \big( \nu_0(\theta) \big) \d\theta
        \\
            &= \lim_{k \to +\infty} \int_{[0, \vartheta_k) \cap \inv{\nu_0}([0, 1])} \varrho^{\smash{G^\prime_\theta}}_v \big(\nu_0(\theta)\big) \d\theta
            + \int_{[0, \vartheta_k) \setminus \inv{\nu_0}([0, 1])} \frac{\varrho^{\smash{G^\prime_\theta}}_v\big(\nu_0(\theta)\big)}{\nu_0(\theta)} \cdot \nu_0(\theta) \d\theta
        \\
            &\le \max_{\substack{G^\prime = (V, A^\prime,  A^*)\\\text{shortest path network}\\\text{with resetting}}} \varrho^{G^\prime}_v(1) \cdot \int_{0}^{\vartheta} \max\set[\big]{1, \nu_0(\theta)} \d\theta
            < +\infty .
        \end{align*}
        Hence, $\hat{\ell} \in \mathcal{L}$ is well-defined and $\ell^{(k)} \preceq \hat{\ell}$ for all $k \in K$.
        Zorn's lemma yields a maximal element $\ell \in \mathcal{L}$.
        \cref{dynamicequilibrium:extension} shows that the domain of $\ell$ has to be $\R$.
        By \cref{thm:differentialequation}, $\ell$ is the earliest times function of a dynamic equilibrium.
    \end{proof}

    \begin{theorem}[Uniqueness of right-monotone dynamic equilibria]
        \label{theorem:dynamicequilibrium:uniqueness}
        Let the inflow rate $\nu_0$ be right-monotone.
        If there are two right-monotone dynamic equilibria, then their earliest times functions agree.
    \end{theorem}

    \begin{proof}
        Let $\Theta \subseteq \R$ be the set on which the earliest times functions of every right-monotone equilibrium with inflow rate $\nu_0$ agree.
        By definition $\Rn \subseteq \Theta$.
        Since earliest times functions are continuous, $\Theta$ is a closed set.
        We need to show that $\Theta = \R$.

        Let $(f^+, f^-)$ be a dynamic equilibrium for inflow rate $\nu_0$ such that $f^+_a$ is right-monotone for all $a = (v, w) \in A$.
        Let $\ell$ be the corresponding earliest times function.
        For almost every $\vartheta \in \R$, $z_a$ is increasing at $\vartheta$ if and only if $f^+_a(\vartheta) > \nu_a$.
        Therefore, $z_a$ is right-monotone as well.
        (Looking into the proof of) \cref{lem:monotone:composition} yields right-monotonicity of $f^+_a \circ \ell_v$ and $z_a \circ \ell_v$ since $\ell_v$ is continuous, hence, locally bounded, and non-decreasing.

        Let $\vartheta \in \Theta$.
        We will see that the support of $\left( f^+_a\big(\ell_v(\vartheta + \varepsilon)\big) \right)_{a = (v, w) \in A}$ and $A^*_{\vartheta + \varepsilon}$ are respectively the same for almost every small enough $\varepsilon > 0$.
        Define the set
        \begin{equation*}
            A^\prime := \set[\big]{a = (v, w) \in A\colon \exists\, \delta > 0\colon f^+_a \circ \ell_v > 0 \text{ a.e. on } (\vartheta, \vartheta + \delta)} .
        \end{equation*}
        For $a \in A \setminus A^\prime$, the right-monotonicity of $f^+_a \circ \ell_v$ implies $f^+_a\big( \ell_v(\vartheta + \varepsilon)\big) = 0$ for almost every small enough $\varepsilon > 0$.
        Thus, the support of $\left( f^+_a\big(\ell_v(\vartheta + \varepsilon)\big) \right)_{a = (v, w) \in A}$ is exactly $A^\prime$.
        Similarly, define the set
        \begin{equation*}
            A^* := \liminf_{\varepsilon \to 0+} A^*_{\vartheta + \varepsilon} = \set[\big]{a = (v, w) \in A\colon \exists\, \delta > 0\colon z_a \circ \ell_v > 0 \text{ on } (\vartheta, \vartheta + \delta)} .
        \end{equation*}
        As $z_a \circ \ell_v$ is right-monotone, $A^* = A^*_{\vartheta + \varepsilon}$ for small enough $\varepsilon > 0$.
        Thus for almost every small enough $\varepsilon > 0$, $\big(f^+_a\big(\ell_v(\vartheta + \varepsilon)\big) \cdot \tfrac{\d \ell_v}{\d \vartheta}(\vartheta + \varepsilon)\big)_{a = (v, w) \in A}$ is a normalized thin flow of value $\nu_0(\vartheta + \varepsilon)$ in $G^\prime := (V, A^\prime, A^*)$ with corresponding labels $\tfrac{\d \ell}{\d \vartheta^+}(\vartheta + \varepsilon)$ .
        Hence, \cref{thm:differentialequation} implies
        \begin{equation*}
            \ell(\vartheta + \varepsilon) = \ell(\vartheta) + \int_{\vartheta}^{\vartheta + \varepsilon} \varrho^{\smash{G^\prime}}\big(\nu_0(\theta)\big) \d \theta .
        \end{equation*}
        Since $(f^+, f^-)$ was chosen arbitrarily, $[\vartheta, \vartheta + \varepsilon] \subseteq \Theta$.
        As a consequence, $\Theta = \R$.
    \end{proof}

    The following example shows that the $\alpha$-extension may fail if $\nu_0$ is not right-monotone.

    \begin{example}[$\alpha$-extension for non-right-monotone inflow rate]
        We consider the graph $G = (V, A)$ which consists only of the source $s$, the sink $t$, and two parallel arcs $a, b$ from $s$ to $t$ with capacity $\nu_a = \nu_{b} = 1$ and transit time $\tau_a = 0$ and $\tau_{b} = 1$, see \cref{fig:example:graph}.
        Note that the network can be transformed to an equivalent instance without multi-arcs by introducing additional vertices on arcs.
        The inflow rate is depicted in \cref{fig:example:inflow} and given by the function
        \begin{equation*}
            \nu_0\colon \R \to \Rp, \quad \vartheta \mapsto \begin{cases}
                0 & \text{for } \vartheta \in (-\infty, 0) \\
                2 & \text{for } \vartheta \in [0, 1] \\
                0 & \text{for } \vartheta \in [1 + 2^{-k-1}, 1 + 2^{-k}), k \in 2 \N + 1 \\
                2 & \text{for } \vartheta \in [1 + 2^{-k-1}, 1 + 2^{-k}), k \in 2 \N \\
                2 & \text{for } \vartheta \in [2, +\infty)
            \end{cases} .
        \end{equation*}
        Note that $\nu_0 \in L^1_\text{loc}(\Rp)$ is not right-monotone at $\vartheta = 1$.
        A dynamic equilibrium is given by the earliest time functions $\ell_s\colon \R \to \R, \vartheta \mapsto \vartheta$ and
        \begin{equation*}
            \ell_t\colon \R \to \R, \quad \vartheta \mapsto \begin{cases}
                \vartheta &\text{for } \vartheta \in (-\infty, 0) \\
                2 \vartheta &\text{for } \vartheta \in [0, 1] \\
                2 + 2^{-k-1} & \text{for } \vartheta \in [1 + 2^{-k-1}, 1 + 2^{-k}), k \in 2 \N - 1 \\
                2 \vartheta - 3 \cdot 2^{-k-2} & \text{for } \vartheta \in [1 + 2^{-k-1}, 1 + 3 \cdot 2^{-k-2}), k \in 2 \N \\
                \vartheta + 1 & \text{for } \vartheta \in [1 + 3 \cdot 2^{-k-2}, 1 + 2^{-k}), k \in 2 \N \\
                \vartheta + 1 &\text{for } \vartheta \in [2, +\infty)
            \end{cases} .
        \end{equation*}
        The graph of $\ell_t$ is shown in \cref{fig:example:label}.
        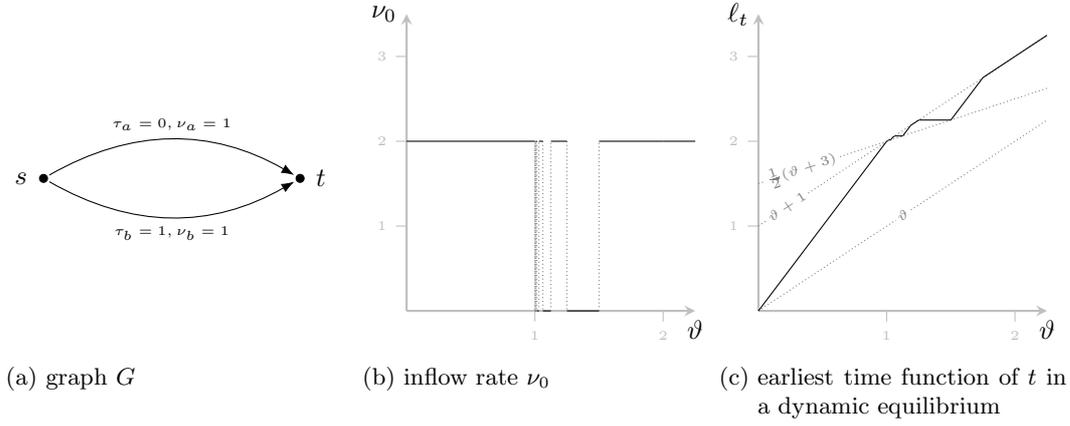
\begin{figure}[t]
            \begin{subfigure}[b]{0.3\textwidth}
                \begin{tikzpicture}[x=48pt,y=32pt]
                    \fill [white] (0, -2) rectangle (2, 2);
                    \node[vertex] (s) at (0,  0) [label={left:$s$}] {};
                    \node[vertex] (t) at (2,  0) [label={right:$t$}] {};
                    \draw[arc] (s) to[out= 30,in= 150] node[above,sloped] {\tiny $\tau_a = 0, \nu_a = 1$} (t);
                    \draw[arc] (s) to[out=-30,in=-150] node[below,sloped] {\tiny $\tau_{b} = 1, \nu_{b} = 1$} (t);
                \end{tikzpicture}
                \caption{graph $G$\\\phantom{y}}
                \label{fig:example:graph}
            \end{subfigure}
            \begin{subfigure}[b]{0.3\textwidth}
                \begin{tikzpicture}[x=48pt, y=32pt]
                    \draw[thick,->,>=stealth,lightgray] (0,0) -- +(2.25, 0) node[below,black] {$\vartheta$};
                    \foreach \t in {1, 2} {\draw[lightgray] (\t, 0) -- +(0, -4pt) node[below] {\tiny\t};};
                    \draw[thick,->,>=stealth,lightgray] (0,0) -- +(0, 3.5) node[left,black] {$\nu_0$};
                    \foreach \t in {1, 2, 3} {\draw[lightgray] (0, \t) -- +(-4pt, 0) node[left] {\tiny\t};};
                    \draw (0,2) -- (1,2);
                    \foreach \j in {0, ..., 4}{
                        \draw[densely dotted,gray] ({1 + pow(2, -2*\j-2)},2) -- ({1 + pow(2, -2*\j-2)},0);
                        \draw ({1 + pow(2, -2*\j-2)},0) -- ({1 + pow(2, -2*\j-1)},0);
                        \draw[densely dotted,gray] ({1 + pow(2, -2*\j-1)},0) -- ({1 + pow(2, -2*\j-1)},2);
                        \draw ({1 + pow(2, -2*\j-1)},2) -- ({1 + pow(2, -2*\j)},2);
                    }
                    \draw (2,2) -- (2.25,2);
                \end{tikzpicture}
                \caption{inflow rate $\nu_0$\\\phantom{y}}
                \label{fig:example:inflow}
            \end{subfigure}
            \begin{subfigure}[b]{0.3\textwidth}
                \begin{tikzpicture}[x=48pt, y=32pt]
                    \draw[thick,->,>=stealth,lightgray] (0,0) -- (2.25, 0) node[below,black] {$\vartheta$};
                    \foreach \t in {1, 2} {\draw[lightgray] (\t, 0) -- +(0, -4pt) node[below] {\tiny\t};};
                    \draw[thick,->,>=stealth,lightgray] (0,0) -- (0, 3.5) node[left,black] {$\ell_t$};
                    \foreach \t in {1, 2, 3} {\draw[lightgray] (0, \t) -- +(-4pt, 0) node[left] {\tiny\t};};
                    \draw[densely dotted,gray] (0,1) -- node[sloped,pos=0.1,gray,rectangle,fill=white,inner sep=1pt] {\tiny $\vartheta + 1$} (2.25,3.25);
                    \draw[densely dotted,gray] (0,1.5) -- node[sloped,pos=0.15,gray,rectangle,fill=white,inner sep=1pt] {\tiny $\tfrac{1}{2} (\vartheta + 3)$} (2.25,2.625);
                    \draw[densely dotted,gray] (0,0) -- node[sloped,pos=0.5,gray,rectangle,fill=white,inner sep=1pt] {\tiny $\vartheta$} (2.25,2.25);
                    \draw (0,0) -- (1,2);
                    \foreach \j in {0, ..., 4}{
                        \draw ({1 +     pow(2, -2*\j-2)}, {2 +     pow(2, -2*\j-2)}) -- ({1 +     pow(2, -2*\j-1)}, {2 +     pow(2, -2*\j-2)});
                        \draw ({1 +     pow(2, -2*\j-1)}, {2 +     pow(2, -2*\j-2)}) -- ({1 + 3 * pow(2, -2*\j-2)}, {2 + 3 * pow(2, -2*\j-2)});
                        \draw ({1 + 3 * pow(2, -2*\j-2)}, {2 + 3 * pow(2, -2*\j-2)}) -- ({1 +     pow(2, -2*\j)},   {2 +     pow(2, -2*\j)});
                    }
                    \draw (2,3) -- (2.25,3.25);
                \end{tikzpicture}
                \caption{earliest time function of $t$ in a dynamic equilibrium}
                \label{fig:example:label}
            \end{subfigure}
            \caption{Example with non-right-monotone inflow rate for which the $\alpha$-extension does not work.}
        \end{figure}
        A queue grows on $a$ up to time $1$, when $b$ becomes active, i.e., $A^\prime_1 = \set{a, b}$.
        After time $1$, the following three phases repeat for every $k \in 2 \N$.
        For times $\vartheta \in [1 + 2^{-k-2}, 1 + 2^{-k-1})$, there is no inflow, the set of active arcs is $A^\prime_\vartheta = \set{a}$, and the queue on $a$ shrinks.
        At time $1 + 2^{-k-1}$ the inflow sets in again.
        For times $\vartheta \in [1 + 2^{-k-1}, 1 + 3 \cdot 2^{-k-2})$, the set of active arcs is still $A^\prime_\vartheta = \set{a}$, but the queue on $a$ is growing.
        At time $1 + 3 \cdot 2^{-k-2}$, $b$ gets active.
        For times $\vartheta \in [1 + 3 \cdot 2^{-k-2}, 1 + 2^{-k})$, the set of active arcs is $A^\prime_\vartheta = \set{a, b}$ and the queue lengths stay constant.
        At time $1 + 2^{-k}$, the inflow stops again, and the three phases repeat.
        In particular, $G^\prime_{\vartheta}$ is not constant on $(1, 1 + \varepsilon)$ for any $\varepsilon > 0$.
\\
        Applying an $\alpha$-extension at time $1$ would yield that for small $\varepsilon > 0$
        \begin{equation*}
            \ell_t\colon \Rp \to \Rp, \quad 1 + \varepsilon \mapsto \begin{cases}
                1 + \tfrac{1}{3} 2^{-k} & \text{for } \varepsilon \in [2^{-k-1}, 2^{-k}), k \in 2 \N - 1 \\
                1 + \varepsilon - \tfrac{1}{3} 2^{-k} & \text{for } \varepsilon \in [2^{-k-1}, 2^{-k}), k \in 2 \N
            \end{cases} .
        \end{equation*}
        Then, $A^\prime_{1 + \varepsilon} = \set{a}$ for small $\varepsilon > 0$.
        But \eqref{eq:differentialequation} is not fulfilled at times $1 + \varepsilon$ for $\varepsilon \in [2^{-k-1}, 2^{-k}), k \in 2 \N$.
    \end{example}

    \section{Normalized thin flows with resetting in series-parallel\\graphs}
    \label{section:seriesparallel}

    The nature of the proof of existence makes the problem of computing a normalized thin flow with resetting part of the complexity class PPAD as introduced by \citet{P1994}.
    It remains open whether this computational problem can be solved in polynomial time for general acyclic arc sets $A^* \subseteq A^\prime \subseteq A$.
    So far, the complexity is known only for the cases $A^* = A^\prime$ and $A^* = \emptyset$.
    For $A^* = A^\prime$, the complementarity problem in \cref{theorem:normalizedthinflow:linearcomplementarityproblem} reduces to a system of linear equations and, hence, can be solved efficiently.
    For $A^* = \emptyset$, the problem is also solvable in polynomial time as shown by \citet{KS2011}.
    Instead of restricting the set $A^*$, we address the computation of normalized thin flows with resetting for graphs $(V, A^\prime)$ that are two-terminal series-parallel (see~\cite{BLS1999}).

    \begin{definition}[Two-terminal directed graph]
        A directed acyclic (multi)graph $G = (V, A)$ is two-terminal if it has a unique source $s_G$ and a unique sink $t_G$, i.e.,
        $\set{v \in V\colon \delta^-(v) = \emptyset} = \set{s_G}$ and
        $\set{v \in V\colon \delta^+(v) = \emptyset} = \set{t_G}$.
    \end{definition}

    \begin{definition}[Two-terminal series-parallel directed graph]
        Let $G_1 = (V_1, A_1)$, $G_2 = (V_2, A_2)$, and $G = (V, A)$ be two-terminal directed graphs.
        $G$ is the \emph{series composition} of $G_1$ and $G_2$ if
        $V_1 \cup V_2 = V$, $V_1 \cap V_2 = \set{t_{G_1}}$, $t_{G_1} = s_{G_2}$, and $A_1 \cupdot A_2 = A$.
        We write $G = G_1 \series G_2$.
        $G$ is the \emph{parallel composition} of $G_1$ and $G_2$ if
        $V_1 \cup V_2 = V$, $V_1 \cap V_2 = \set{s_{G_1}, t_{G_1}}$, $s_{G_1} = s_{G_2}$, $t_{G_1} = t_{G_2}$, and $A_1 \cupdot A_2 = A$.
        We write $G = G_1 \parallel G_2$.
        The class of \emph{two-terminal series-parallel directed graphs} is the smallest set of graphs which contains the graph(s) with two vertices and a single arc in-between, and is closed with respect to series and parallel composition.
    \end{definition}

    The following two lemmas show how the function $\varrho^{G^\prime}$ relates to the functions $\varrho^{\smash{G^\prime_1}}$ and $\varrho^{\smash{G^\prime_2}}$ for the series composition $G^\prime = G^\prime_1 \series G^\prime_2$ as well as the parallel composition $G^\prime = G^\prime_1 \parallel G^\prime_2$.
    Restricting a thin flow with resetting in $G^\prime$ to $G^\prime_1$ or $G^\prime_2$ yields a thin flow with resetting in the respective graph.
    The other way round, the conditions for a flow to be a thin flow with resetting in $G^\prime$ are roughly a combination of the conditions for flows in $G^\prime_1$ and $G^\prime_2$.
    The only additional requirement is to have one common label $\ell^\prime_v$ at every common vertex $v$ of $G^\prime_1$ and $G^\prime_2$.
    For the series composition this synchronization of labels is straightforward.
    For the parallel composition, it is achieved by splitting the total amount of flow appropriately between $G^\prime_1$ or $G^\prime_2$.

    \begin{lemma}[Series composition]
        \label{lem:series}
        Let $G^\prime = (V, A^\prime)$ be a two-terminal directed acyclic graph.
        Further, let $V_1, V_2 \subseteq V$ be vertex sets and set $G^\prime_i := G^\prime[V_i]$ for $i = 1, 2$.
        If $G^\prime = G^\prime_1 * G^\prime_2$ and $n_1, n_2, n$ are the numbers of breakpoints of $\varrho^{G^\prime_1}, \varrho^{G^\prime_2}, \varrho^{G^\prime}$, respectively, then $n \le n_1 + n_2$.
    \end{lemma}

    \begin{proof}
        Let $r \in V$ be such that $V_1 \cap V_2 = \set{r}$.
        As discussed in the paragraph before \cref{lem:series}, the conditions of \cref{definition:normalizedthinflow} for $G^\prime_1$ and $G^\prime_2$  also appear for $G^\prime$.
        Applying \cref{lem:normalizedthinflow:parametric:monotonicity} yields that, for $\nu^\prime_0 > 0$
        \begin{alignat*}{2}
            \varrho^{G^\prime}_v(\nu^\prime_0) &= \varrho^{\smash{G^\prime_1}}_v(\nu^\prime_0) &\quad&\text{if } v \in V_1 \text{ and}
        \\
            \varrho^{G^\prime}_v(\nu^\prime_0) &= \varrho^{\smash{G^\prime_1}}_r(\nu^\prime_0) \cdot \varrho^{\smash{G^\prime_2}}_v\left( \frac{\nu^\prime_0}{\varrho^{\vphantom{G^\prime}\smash{G^\prime_1}}_r(\nu^\prime_0)} \right) &\quad&\text{if } v \in V_2 .
        \end{alignat*}
        For $v \in V_2$, $\varrho^{G^\prime}_v$ can have a breakpoint at $\nu^\prime_0 > 0$ only if $\varrho^{\smash{G^\prime_1}}_r$ does, or $\varrho^{\smash{G^\prime_2}}_v$ has a breakpoint at $\frac{\nu^\prime_0}{\varrho^{\vphantom{G^\prime}\smash{G^\prime_1}}_r(\nu^\prime_0)}$ and the function $\frac{\nu^\prime_0}{\varrho^{\vphantom{G^\prime}\smash{G^\prime_1}}_r(\nu^\prime_0)}$ is not constant around $\nu^\prime_0$.
        \Cref{lem:normalizedthinflow:parametric:monotonicity} shows that $\frac{\nu^\prime_0}{\varrho^{\vphantom{G^\prime}\smash{G^\prime_1}}_r(\nu^\prime_0)}$ is monotone in $\nu^\prime_0$.
        Therefore, $n \le n_1 + n_2$ follows.
        See \cref{fig:series} for an illustration.
    \end{proof}

    \begin{lemma}[Parallel composition]
        \label{lem:parallel}
        Let $G^\prime = (V, A^\prime)$ be a two-terminal directed acyclic graph.
        Further, let $V_1, V_2 \subseteq V$ be vertex sets and set $G^\prime_i := G^\prime[V_i]$ for $i = 1, 2$.
        If $G^\prime = G^\prime_1 \parallel G^\prime_2$ and $n_1, n_2, n$ are the numbers of breakpoints of $\varrho^{G^\prime_1}, \varrho^{G^\prime_2}, \varrho^{G^\prime}$, respectively, then $n \le n_1 + n_2 + 1$.
    \end{lemma}

    \begin{proof}
        Considering \cref{definition:normalizedthinflow} or the complementarity problem from \cref{theorem:normalizedthinflow:linearcomplementarityproblem} reveals that for every $\nu^\prime_0 \ge 0$
        \begin{alignat*}{2}
            \varrho^{\smash{G^\prime}}_t(\nu^\prime_0) &= \min_{i \in \set{1, 2}} \varrho^{\smash{G^\prime_i}}_t(\nu^\prime_i) \\
            \varrho^{\smash{G^\prime}}_t(\nu^\prime_0) &= \varrho^{\smash{G^\prime_i}}_t(\nu^\prime_i) &\quad&\text{for } i \in \set{1, 2} \text{ with } \nu^\prime_i > 0 \\
            \varrho^{\smash{G^\prime}}_v(\nu^\prime_0) &= \varrho^{\smash{G^\prime_i}}_v(\nu^\prime_i) &\quad&\text{for } i \in \set{1, 2}, v \in V_i \setminus \set{t} \\
            \nu^\prime_1 + \nu^\prime_2 &= \nu^\prime_0 \\
            \nu^\prime_1, \nu^\prime_2 &\ge 0 .
        \end{alignat*}
        By \cref{corollary:normalizedthinflow:uniqueness}, $\varrho^{\smash{G^\prime}}_t(\nu^\prime_0)$ is uniquely determined by this non-linear complementarity problem for given $\nu^\prime_0$, $\varrho^{\smash{G^\prime_1}}_t$, and $\varrho^{\smash{G^\prime_2}}_t$.
        However, $\nu^\prime_1$ and $\nu^\prime_2$ are generally not uniquely determined.
        This fact is independent of the equations for all $v \in V \setminus \set{t}$.
        Hence, we ignore those for the time being and find functions $\nu_1\colon \Rp \to \Rp$ and $\nu_2\colon \Rp \to \Rp$ which are piecewise linear and describe a solution $\nu_1(\nu^\prime_0), \nu_2(\nu^\prime_0)$ for every $\nu^\prime_0 \ge 0$.
        The idea is to define them on a discrete set for which they are unique solutions and interpolate linearly.

        For $i \in \set{1, 2}$, define the functions $\ubar{\nu}_i\colon \Rp \to \Rp$ and $\bar{\nu}_i\colon \Rp \to \Rp$ by
        \begin{equation*}
            \ubar{\nu}_i(\rho) := \min \set[\big]{\nu^\prime_i \ge 0 \colon \varrho^{\smash{G^\prime_i}}_t(\nu^\prime_i) \ge \rho} \quad \text{ and } \quad
            \bar{\nu}_i(\rho) := \inf \set[\big]{\nu^\prime_i \ge 0 \colon \varrho^{\smash{G^\prime_i}}_t(\nu^\prime_i) > \rho } .
        \end{equation*}
        Then for $\rho \ge \varrho^{\smash{G^\prime_i}}_t(0)$, the preimage of $\rho$ under the function $\varrho^{\smash{G^\prime_i}}_t$ is exactly the interval $\left[\ubar{\nu}_i(\rho), \bar{\nu}_i(\rho)\right]$.
        For $\rho < \varrho^{\smash{G^\prime_i}}_t(0)$, we get $\ubar{\nu}_i(\rho) = \bar{\nu}_i(\rho) = 0$.
        Note that $\bar{\nu}_i(\rho) < \ubar{\nu}_i(\hat{\rho})$ for all $\varrho^{\smash{G^\prime_i}}_t(0) \le \rho < \hat{\rho}$ due to the monotonicity of $\varrho^{\smash{G^\prime_i}}_t$.
        \\
        Without loss of generality, we can assume that $\varrho^{\smash{G^\prime_1}}_t(0) \le \varrho^{\smash{G^\prime_2}}_t(0)$.
        For $\nu_0^\prime \ge 0$, there is a unique $\rho \ge \varrho^{\smash{G^\prime_1}}_t(0)$ such that $\nu_0^\prime \in \big[ \ubar{\nu}_1(\rho) + \ubar{\nu}_2(\rho), \bar{\nu}_1(\rho) + \bar{\nu}_2(\rho) \big]$.
        If $\rho < \varrho^{\smash{G^\prime_2}}_t(0)$, then $\rho = \varrho^{\smash{G^\prime_1}}_t(\nu_0^\prime)$.
        Otherwise, $\rho$ is the value at the intersection of the graphs of $\varrho^{\smash{G^\prime_1}}_t$ and $\varrho^{\smash{G^\prime_2}}_t(\nu_0^\prime - \cdot)$, see \cref{fig:parallel}.
        The set of solutions $\big(\varrho^{\smash{G^\prime}}_t(\nu^\prime_0), \nu_1^\prime, \nu_2^\prime \big)$ to the above complementarity problem can be written as
        \begin{equation*}
            \set[\big]{(\rho, \nu_1^\prime, \nu_2^\prime) \in \Rp \times \left[\ubar{\nu}_1(\rho), \bar{\nu}_1(\rho)\right] \times \left[\ubar{\nu}_2(\rho), \bar{\nu}_2(\rho)\right]\colon \nu^\prime_1 + \nu^\prime_2 = \nu^\prime_0, \rho \ge \varrho^{\smash{G^\prime_1}}_t(0) } .
        \end{equation*}
        Thus, a solution $(\rho, \nu_1^\prime, \nu_2^\prime)$ is unique if and only if $\nu^\prime_0 = \ubar{\nu}_1(\rho) + \ubar{\nu}_2(\rho)$ or $\nu^\prime_0 = \bar{\nu}_1(\rho) + \bar{\nu}_2(\rho)$.

        Let $R := \set[\big]{\varrho^{\smash{G^\prime_1}}_t(0), \varrho^{\smash{G^\prime_2}}_t(0)} \cup \set[\big]{\varrho^{\smash{G^\prime_i}}_t(\nu^\prime_0) \colon i \in \set{1, 2} \text{ and } \nu^\prime_0 \ge 0 \text{ is breakpoint of } \varrho^{\smash{G^\prime_i}}_t }$ be the values at breakpoints of $\varrho^{\smash{G^\prime_1}}_t$ and $\varrho^{\smash{G^\prime_2}}_t$ (including the border of the domain).
        By the above, $\nu_1$ and $\nu_2$ are in particular uniquely determined on $\Nu := \set{\ubar{\nu}_1(\rho) + \ubar{\nu}_2(\rho)\colon \rho \in R} \cup \set{\bar{\nu}_1(\rho) + \bar{\nu}_2(\rho)\colon \rho \in R}$.
        Hence, for every $i \in \set{1, 2}$ and $\rho \in R$, we set
        \begin{equation*}
            \nu_i\big( \ubar{\nu}_1(\rho) + \ubar{\nu}_2(\rho) \big) := \ubar{\nu}_i(\rho)
            \quad \text{and} \quad
            \nu_i\big( \bar{\nu}_1(\rho) + \bar{\nu}_2(\rho) \big) := \bar{\nu}_i(\rho) .
        \end{equation*}
        These definitions are extended to $\Rp$ by linear interpolation.
        Then for $i = 1, 2$, the composition $\varrho^{\smash{G^\prime_i}}_t \circ \nu_i$ is piecewise linear with breakpoints only in $\Nu$.
        In particular, all breakpoints of $\varrho^{\smash{G^\prime}}_t$ lie in~$\Nu$.
        By construction, $\nu_1(\nu^\prime_0)$ and $\nu_2(\nu^\prime_0)$ define a solution for every $\nu^\prime_0 \in N$.
        Linearity in-between the breakpoints generalizes this to all $\nu^\prime_0 \in \Rp$.

        The above allows to bound the number of breakpoints of $\varrho^{G^\prime}_t$.
        For $\rho \in R$, the strict inequality $\ubar{\nu}_1(\rho) + \ubar{\nu}_2(\rho) < \bar{\nu}_1(\rho) + \bar{\nu}_2(\rho)$ holds only if there is $i \in \set{1, 2}$ such that $\ubar{\nu}_i(\rho) < \bar{\nu}_i(\rho)$ and, thus, $\ubar{\nu}_i(\rho)$ and $\bar{\nu}_i(\rho)$ are two breakpoints of $\varrho^{\smash{G^\prime_i}}_t$ with value $\rho$.
        This shows the inequality $\card{\Nu} \le n_1 + n_2 + 2$, where the constant is accounting for the border of the domain.
        As $0 \in \Nu$ due to $\varrho^{\smash{G^\prime_1}}_t(0) \in R$, it follows that $\varrho^{G^\prime}_t$ has at most $n_1 + n_2 + 1$ breakpoints.
        Repeating the counting more carefully allows to extend this bound to the number of breakpoints of $\varrho^{G^\prime}$.
        The set $\Nu$ only contains breakpoints which are based on breakpoints of $\varrho^{\smash{G^\prime_1}}_t$ and $\varrho^{\smash{G^\prime_2}}_t$.
        For $i = 1, 2$, any breakpoint of $\varrho^{\smash{G^\prime_i}}$ at which $\varrho^{\smash{G^\prime_i}}_t$ is differentiable leads to at most one breakpoint of $\varrho^{\smash{G^\prime_i}} \circ \nu_i$ additional to those in $\Nu$.
        Consequently, $\varrho^{G^\prime}$ does not have more than $n_1 + n_2 + 1$ breakpoints.
    \end{proof}

    \begin{figure}
        \centering
        \begin{subfigure}[b]{\textwidth}
            \captionsetup{width=\textwidth}
            \begin{tikzpicture}[x=3.4pt, y=16pt,breakpoint/.style={circle,fill=black,inner sep=0pt,minimum width=2pt}]
                \foreach \y in {1,2,2.3,2.5,2.6} {\draw[densely dotted,lightgray] (0, \y) -- (118, \y);};
                \begin{scope}[shift={(0, 0)}]
                    \draw[thick,->,>=stealth,lightgray] (0,0) -- +(24, 0) node[below=2pt,black] {$\nu^\prime_0$};
                    \draw[thick,->,>=stealth,lightgray] (0,0) -- +(0, 4) node[left=-2pt,black,anchor=north east] {$\varrho^{\smash{G^\prime_1}}_r$};
                    \draw[densely dotted,gray] (15, 0) node[below] {\small$\hat{\nu}_0$} -- ++(0, 4);
                    \coordinate (b0) at (0, 0);
                    \coordinate[breakpoint] (b1) at (4, 1);
                    \coordinate[breakpoint] (b2) at (6, 1);
                    \coordinate[breakpoint] (b3) at (12, 2);
                    \coordinate[breakpoint] (b4) at (17, 2.5);
                    \coordinate (b5) at (24, 3.5);
                    \draw (b0) -- (b1) -- (b2) -- (b3) -- (b4) -- (b5);
                \end{scope}
                \begin{scope}[shift={(35, 0)}]
                    \draw[thick,->,>=stealth,lightgray] (0,0) -- +(24, 0) node[below=2pt,black] {$\nu^\prime_0$};
                    \draw[thick,->,>=stealth,lightgray] (0,0) -- +(0, 4) node[left=-2pt,black,anchor=north east] {$\varrho^{\smash{G^\prime_2}}_t$};
                    \draw[densely dotted,gray] (6.5, 0) node[below] {\small$\tfrac{\hat{\nu}_0}{\varrho^{\vphantom{G^\prime}\smash{G^\prime_1}}_r(\hat{\nu}_0)}$} -- ++(0, 4);
                    \coordinate (b0) at (0, 1);
                    \coordinate[breakpoint] (b1) at (6.5, 1);
                    \coordinate (b2) at (24, 3.7);
                    \draw (b0) -- (b1) -- (b2);
                \end{scope}
                \begin{scope}[shift={(70, 0)}]
                    \draw[thick,->,>=stealth,lightgray] (0,0) -- +(48, 0) node[below=2pt,black] {$\nu^\prime_0$};
                    \draw[thick,->,>=stealth,lightgray] (0,0) -- +(0, 4) node[left=-3pt,black,anchor=north east] {$\varrho^{\smash{G^\prime}}_t$};
                    \draw[densely dotted,gray] (15, 0) node[below] {\small$\hat{\nu}_0$} -- ++(0, 4);
                    \coordinate (b0) at (0, 0);
                    \coordinate[breakpoint] (b1) at (4, 1);
                    \coordinate[breakpoint] (b2) at (6, 1);
                    \coordinate[breakpoint] (b3) at (12, 2);
                    \coordinate[breakpoint] (b4) at (15, 2.3);
                    \coordinate[breakpoint] (b5) at (17, 2.6);
                    \coordinate (b6) at (24, 3.7);
                    \draw (b0) -- (b1) -- (b2) -- (b3) -- (b4) -- (b5) -- (b6);
                \end{scope}
            \end{tikzpicture}
            \caption{Series composition $G^\prime = G^\prime_1 \series G^\prime_2$ at common vertex $r$. The breakpoint of $\varrho^{\smash{G^\prime}}_t$ at $\hat{\nu}_0$ results from the breakpoint of $\varrho^{\smash{G^\prime_2}}_t$.}
            \label{fig:series}
        \end{subfigure}
        \par\medskip
        \begin{subfigure}[b]{\textwidth}
            \captionsetup{width=\textwidth}
            \begin{tikzpicture}[x=3.4pt, y=16pt,breakpoint/.style={circle,fill=black,inner sep=0pt,minimum width=2pt}]
                \foreach \y in {1,2,2.5,4} {\draw[densely dotted,lightgray] (0, \y) -- (118, \y);};
                \node[left] at (0, 2.5) {\small$\rho$};
                \begin{scope}[shift={(0, 0)}]
                    \draw[thick,->,>=stealth,lightgray] (0,0) -- +(24, 0) node[below=2pt,black] {$\nu^\prime_1$};
                    \draw[thick,->,>=stealth,lightgray] (0,0) -- +(0, 6) node[left=-2pt,black,anchor=north east] {$\varrho^{\smash{G^\prime_1}}_t$};
                    \draw[densely dotted,gray] (17, 0) node[below] {\small$\nu_1(\hat{\nu}_0)$} -- (17, 6);
                    \coordinate (b0) at (0, 0);
                    \coordinate[breakpoint] (b1) at (4, 1);
                    \coordinate[breakpoint] (b2) at (6, 1);
                    \coordinate[breakpoint] (b3) at (12, 2);
                    \coordinate[breakpoint] (b4) at (17, 2.5);
                    \coordinate (b5) at (24, 3.5);
                    \draw (b0) -- (b1) -- (b2) -- (b3) -- (b4) -- (b5);
                \end{scope}
                \begin{scope}[shift={(35, 0)}]
                    \draw[thick,->,>=stealth,lightgray] (0,0) -- +(24, 0) node[below=2pt,black] {$\nu^\prime_2$};
                    \draw[thick,->,>=stealth,lightgray] (0,0) -- +(0, 6) node[left=-2pt,black,anchor=north east] {$\varrho^{\smash{G^\prime_2}}_t$};
                    \draw[densely dotted,gray] (8.5, 0) node[below] {\small$\nu_2(\hat{\nu}_0)$} -- (8.5, 6);
                    \coordinate (b0) at (0, 0);
                    \coordinate[breakpoint] (b1) at (6, 2);
                    \coordinate[breakpoint] (b2) at (16, 4);
                    \coordinate (b3) at (24, 6);
                    \draw (b0) -- (b1) -- (b2) -- (b3);
                \end{scope}
                \begin{scope}[shift={(70, 0)}]
                    \draw[thick,->,>=stealth,lightgray] (0,0) -- +(48, 0) node[below=2pt,black] {$\nu^\prime_0$};
                    \draw[thick,->,>=stealth,lightgray] (0,0) -- +(0, 6) node[left=-3pt,black,anchor=north east] {$\varrho^{\smash{G^\prime}}_t$};
                    \foreach \x in {17,25.5} {\draw[densely dotted,lightgray] (\x, 0) -- (\x, 6);};
                    \begin{scope}[shift={(0, 0)},gray]
                        \coordinate (b0) at (0, 0);
                        \coordinate (b1) at (4, 1);
                        \coordinate (b2) at (6, 1);
                        \coordinate (b3) at (12, 2);
                        \coordinate (b4) at (17, 2.5);
                        \coordinate (b5) at (41.5, 6);
                        \draw (b0) -- (b1) -- (b2) -- (b3) -- (b4) -- node[midway,sloped,xshift=16pt,yshift=8pt] {\tiny$\varrho^{\smash{G^\prime_1}}_t(\nu^\prime_0)$} (b5);
                    \end{scope}
                    \begin{scope}[gray]
                        \coordinate (b0) at ({25.5 - 0}, 0);
                        \coordinate (b1) at ({25.5 - 6}, 2);
                        \coordinate (b2) at ({25.5 - 16}, 4);
                        \coordinate (b3) at ({25.5 - 24}, 6);
                        \draw (b0) node[below] {\small$\hat{\nu}_0$} -- (b1) -- node[midway,sloped,xshift=-20pt,yshift=8pt] {\tiny$\varrho^{\smash{G^\prime_2}}_t(\hat{\nu}_0 - \nu^\prime_0)$} (b2) -- (b3);
                    \end{scope}
                    \coordinate (b0) at (0, 0);
                    \coordinate[breakpoint] (b1) at (7, 1);
                    \coordinate[breakpoint] (b2) at (9, 1);
                    \coordinate[breakpoint] (b3) at (18, 2);
                    \coordinate[breakpoint] (b4) at (25.5, 2.5);
                    \coordinate[breakpoint] (b5) at (43.3, 4);
                    \coordinate (b6) at (48, 4.44);
                    \draw (b0) -- (b1) -- (b2) -- (b3) -- (b4) -- (b5) -- (b6);
                    \draw [decorate,decoration={brace,amplitude=4pt,raise=0,mirror}] (25.5,5.5) -- (17,5.5) node[black,midway,yshift=10pt] {\tiny$\nu_2(\hat{\nu}_0)$};
                    \draw [decorate,decoration={brace,amplitude=4pt,raise=0,mirror}] (17,5.5) -- (0,5.5) node[black,midway,yshift=10pt] {\tiny$\nu_1(\hat{\nu}_0)$};
                \end{scope}
            \end{tikzpicture}
            \caption{Parallel composition $G^\prime = G^\prime_1 \parallel G^\prime_2$. The breakpoint $\hat{\nu}_0$ of $\varrho^{\smash{G^\prime}}_t$ is based on the intersection of $\varrho^{\smash{G^\prime_1}}_t(\nu^\prime_0)$ and $\varrho^{\smash{G^\prime_2}}_t(\hat{\nu}_0 - \nu^\prime_0)$ at a breakpoint of $\varrho^{\smash{G^\prime_1}}_t$.}
            \label{fig:parallel}
        \end{subfigure}
        \caption{Examples for a series and a parallel composition.}
        \label{fig:seriesparallel}
    \end{figure}
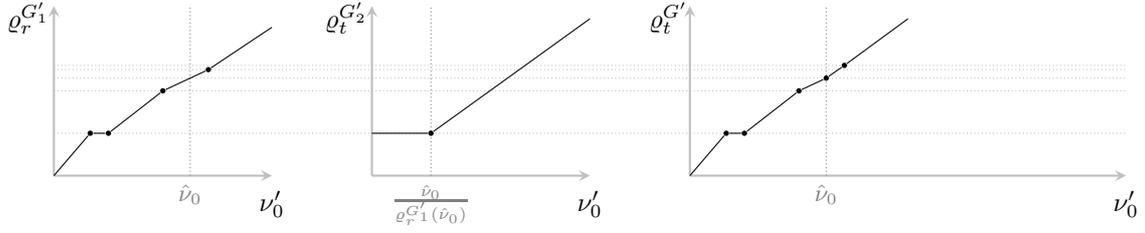
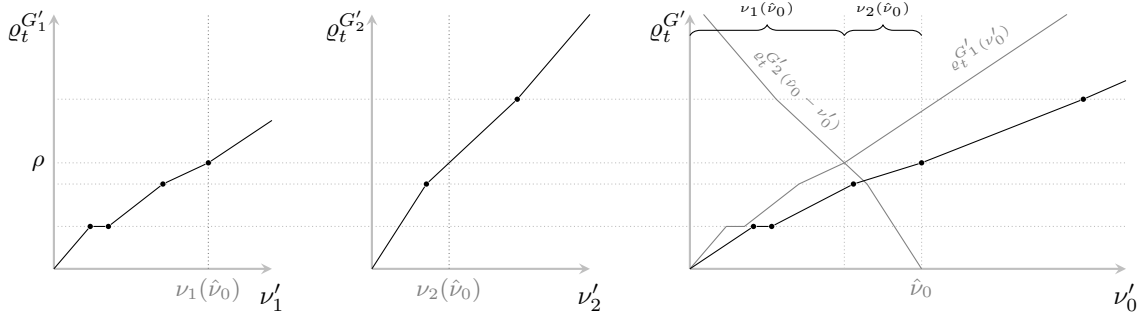

    Understanding both compositions that series-parallel graphs are based on allows us to compute the function $\varrho^{G^\prime}$ for these graphs recursively.
    Our algorithm not only computes a normalized thin flow with resetting for one single flow value, but for all flow values simultaneously.
    As seen in \cref{section:evolution}, this is necessary in order to compute dynamic equilibria for more general inflow rates than piecewise constant functions.

    \begin{theorem}[Labels in series-parallel graphs]
        \label{theorem:thinflow:twoterminalgraph}
        Let $G^\prime := (V, A^\prime)$ be a two-terminal series-parallel directed graph and $A^* \subseteq A^\prime$.
        Then $\varrho^{G^\prime}$ has at most $2 \card{A^\prime} - \card{A^*} - \card{V} + 1$ many breakpoints and can be computed in polynomial time.
    \end{theorem}

    \begin{proof}
        We prove the claim by a structural induction on two-terminal series-parallel graphs.
        For the base case, assume $V = \set{s, t}$ and $A^\prime$ contains a single arc $a = (s, t)$.
        Then, $\varrho^{G^\prime}_s \equiv 1$ and $\varrho^{G^\prime}_t \equiv \varrho^a(1, \cdot)$ as defined in \cref{definition:normalizedthinflow}.
        Hence, $\varrho^{G^\prime}$ has exactly $1 - \card{A^*} = 2 \card{A^\prime} - \card{A^*} - \card{V} + 1$ many breakpoints.
        Assume there are $V_1, V_2 \subseteq V$ such that $G^\prime = G^\prime[V_1] \series G^\prime[V_2]$.
        The statement for $G^\prime$ follows immediately from the induction hypothesis and \cref{lem:series} as $\card{V} = \card{V_1} + \card{V_2} - 1$.
        Now, assume there are $V_1, V_2 \subseteq V$ such that $G^\prime = G^\prime[V_1] \parallel G^\prime[V_2]$.
        Then, the statement for $G^\prime$ follows immediately from the induction hypothesis and \cref{lem:parallel} as $\card{V} = \card{V_1} + \card{V_2} - 2$.

        Since all involved functions are piecewise linear with few breakpoints, they can be represented efficiently by their linear pieces.
        Regarding the proofs of \cref{lem:series,lem:parallel}, it becomes evident that in both cases $\varrho^{G^\prime}$ can be constructed from $\varrho^{G^\prime_1}$ and $\varrho^{G^\prime_2}$ efficiently.
        To compose $G^\prime$ from single arcs, $\card{A^\prime} - 1$ compositions are needed.
        In total, the function $\varrho^{G^\prime}$ can be computed in polynomial time in the size of the input $G^\prime$, $A^*$, and $(\nu_a)_{a \in A^\prime}$.
    \end{proof}

    If the corresponding labels $\ell^\prime$ of a normalized thin flow with resetting are known, computing flow values is not hard.
    The flow value $x^\prime_a$ on $a = (v, w) \in A^\prime$ lies in the interval $[0, \nu_a \ell^\prime_w]$.
    If $\ell^\prime_v > \ell^\prime_w$ and $a \not\in A^*$, then $x^\prime_a$ is zero.
    If $\ell^\prime_v < \ell^\prime_w$ or $a \in A^*$, then $x^\prime_a = \nu_a \ell^\prime_w$ holds.
    Finding a flow satisfying these conditions can be done by a simple flow computation.

    \begin{corollary}[Normalized thin flows in series-parallel graphs]
        Normalized thin flows with resetting in two-terminal series-parallel graphs can be computed in polynomial time.
    \end{corollary}

    \section{Conclusion and outlook}
    \label{section:conclusion}

    We examine normalized thin flows with resetting as a parametric problem in dependency on the flow value.
    The results allow us to give a constructive proof for the existence of dynamic equilibria for single-source single-sink networks with right-monotone inflow rate.
    Further, we obtain a polynomial time algorithm for computing thin flows with resetting on two-terminal series-parallel networks.
    The recursive approach that we take for the latter does not seem to generalize to arbitrary networks.
    A central aspect that we use in this recursion is that every considered subnetwork has a single source and a single sink.

    Major open questions on the model, like the Price of Anarchy and the number of phases, seem to require insights into the relation between the thin flows of subsequent phases.
    These, in turn, would need a better understanding of the dependency of thin flows with resetting on the sets of active and resetting arcs.

    The characterization of normalized thin flows with resetting by a linear complementarity problem in \cref{theorem:normalizedthinflow:linearcomplementarityproblem} allows us to deduce some basic properties of the functions $\varrho^{G^\prime}$.
    Beyond this, it opens the way to approach normalized thin flows with resetting through the existing machinery of linear complementarity problems.
    The properties of the presented linear complementarity problem as established in \cref{lemma:linearcomplementarityproblem:p0} and in the proof of \cref{thm:normalizedthinflow:existence} show that normalized thin flows with resetting can be computed in finitely many steps via Lemke's algorithm.
    While there is no immediate non-trivial bound on the number of steps in theory, it suggests an efficient method in practice.

    \section*{Acknowledgments}
    The author is grateful to Jos{\'e} Correa, Andr{\'e}s Cristi, Dario Frascaria, Neil Olver, Tim Oosterwijk, Leon Sering, Laura Vargas Koch, and Philipp Warode for fruitful discussions on the topic.
    He further would like to thank Ulf Friedrich for his help on measure theoretic aspects and Andreas S. Schulz for his helpful input.
    Finally, the author thanks two anonymous reviewers for their valuable comments which has led to an improved presentation of this article.

    \printbibliography

\end{document}